\documentclass[a4paper,11pt,noamsfonts]{amsart}
\pdfoutput=1
\usepackage[utf8]{inputenc}
\usepackage[english]{babel}
\usepackage[T1]{fontenc}

\usepackage{newtx}

\usepackage{microtype}
\usepackage{color}
\usepackage{pgfplots}
\pgfplotsset{compat=1.18}
\usepackage{tikz}
\usetikzlibrary{quantikz2}
\usetikzlibrary{cd,pgfplots.statistics}

\usepackage{standalone}
\usepackage[numbers]{natbib}
\usepackage[hidelinks]{hyperref}
\usepackage{bbold}
\usepackage{bm}
\usepackage{caption}
\usepackage{subcaption}
\usepackage[a4paper, total={150mm,215mm}, centering]{geometry}

\newtheorem{theorem}{Theorem}

\newtheorem{corollary}[theorem]{Corollary}
\newtheorem{lemma}[theorem]{Lemma}
\theoremstyle{definition}
\newtheorem{definition}[theorem]{Definition}
\theoremstyle{remark}
\newtheorem{example}[theorem]{Example}
\newtheorem{remark}[theorem]{Remark}

\definecolor{mycolor1}{rgb}{0.00000,0.44700,0.74100}% blau
\definecolor{mycolor2}{rgb}{0.85000,0.32500,0.09800}% rot/orange
\definecolor{mycolor4}{rgb}{0.92900,0.69400,0.12500}% gelb

\DeclareFontFamily{U}{matha}{\hyphenchar\font45}
\DeclareFontShape{U}{matha}{m}{n}{
	<-6> matha5 <6-7> matha6 <7-8> matha7
	<8-9> matha8 <9-10> matha9
	<10-12> matha10 <12-> matha12
}{}
\DeclareSymbolFont{matha}{U}{matha}{m}{n}

\DeclareFontFamily{U}{mathx}{\hyphenchar\font45}
\DeclareFontShape{U}{mathx}{m}{n}{
	<-6> mathx5 <6-7> mathx6 <7-8> mathx7
	<8-9> mathx8 <9-10> mathx9
	<10-12> mathx10 <12-> mathx12
}{}
\DeclareSymbolFont{mathx}{U}{mathx}{m}{n}

\DeclareMathDelimiter{\vvvert} {0}{matha}{"7E}{mathx}{"17}

\newcommand{\N}{\mathbb{N}}

\newcommand{\C}{\mathbb{C}}

\newcommand{\ie}{\eta}
\newcommand{\Cie}{C_\ie}

\newcommand{\qc}{c}

\DeclareMathOperator{\polylog}{polylog}

\DeclareMathOperator{\Id}{Id}

\newcommand{\B}{\mathcal{B}}

\newcommand{\A}{\mathcal{A}}
\newcommand{\F}{\mathcal{F}}
\newcommand{\G}{\mathcal{G}}

\newcommand{\Diff}{D}
\newcommand{\tol}{\varepsilon}
\newcommand{\CNOT}[1]{\operatorname{C}_{#1}\!\operatorname{NOT}}

\let\oldket\ket
\newcommand{\Cg}{C}
 
\renewcommand{\ket}[1]{|#1\rangle}
\allowdisplaybreaks
\makeatletter
\def\paragraph{\@startsection{paragraph}{4}%
  {\z@}{2ex}{-1em}%
  {\normalfont\bfseries}}
\makeatother

\makeatletter
\renewcommand{\qwbundle}[2][]{%
  \pgfkeys{/quantikz/gates/.cd,style=,Strike Width=0.08cm,Strike Height=0.12cm,#1}%
  \pgfkeysgetvalue{/quantikz/gates/style}{\qz@style}%
  \pgfkeysgetvalue{/quantikz/gates/Strike Width}{\qz@sw}%
  \pgfkeysgetvalue{/quantikz/gates/Strike Height}{\qz@sh}%
  \expanded{%
    \noexpand\arrow[strike arrow={\qz@sw}{\qz@sh}{\unexpanded{#2}},\qz@style,phantom]{l}%
  }%
}
\makeatother

\begin{document}

\title{Nonlinear quantum computation by amplified encodings}
\author[Matthias Deiml, Daniel Peterseim]{Matthias Deiml$^1$, Daniel Peterseim$^{1,2}$}
\address{$^1$Institute of Mathematics, University of Augsburg, Universit\"atsstr.~12a, 86159 Augsburg, Germany}
\address{$^2$Centre for Advanced Analytics and Predictive Sciences (CAAPS), University of Augsburg, Universit\"atsstr.~12a, 86159 Augsburg, Germany}
\email{\{matthias.deiml,daniel.peterseim\}@uni-a.de}
\date{November 25, 2024}
% 68Q12 65H10 81P68 65J15 65N22

\begin{abstract}
This paper presents a novel framework for high-dimensional nonlinear quantum computation that exploits tensor products of amplified vector and matrix encodings to efficiently evaluate multivariate polynomials. The approach enables the solution of nonlinear equations by quantum implementations of the fixed-point iteration and Newton's method, with quantitative runtime bounds derived in terms of the error tolerance. These results show that a quantum advantage, characterized by a logarithmic scaling of complexity with the dimension of the problem, is preserved. While Newton’s method attains near-optimal theoretical complexity, the fixed-point iteration may be better suited to near-term noisy hardware, as supported by our numerical experiments.
\end{abstract}

\maketitle


\section{Introduction}
This paper introduces a novel framework to efficiently evaluate multivariate polynomials $f \colon \C^N \to \C^N$ in potentially very high dimensions $N \in \N$ on a quantum computer. This approach extends to any nonlinear function $f$ for which an efficient approximation as a multivariate polynomial is available. The framework is particularly suitable for solving nonlinear equations of the form
\begin{equation} \label{eq:nonlinear-eq}
f(x^*) = 0.
\end{equation}
Using Newton's method, a quantum solver with a favorable complexity can be constructed. 
\begin{theorem}[Main result: Quantum nonlinear solver for polynomials] \label{thm:main}
    There is an algorithm that takes as input block encodings of the coefficients of a multivariate polynomial $f$ as well as a circuit that prepares a state $\ket{x^{(0)}}$ and outputs a given quantity of interest of the solution $x^*$ of \eqref{eq:nonlinear-eq} with arbitrarily small error tolerance $\tol > 0$ and failure probability $\delta > 0$ in runtime
    \[
    \mathcal{O}(\tol^{-(1+\mu)}\log \delta^{-1})
    \]
    for any $\mu >0$, provided that the classical Newton's method converges quadratically for the initial guess $x^{(0)}$ in exact arithmetic.
\end{theorem}
The theorem is constructive, and its proof is given in Section~\ref{ss:newton}.
The constant in the $\mathcal{O}$-notation for the runtime bound depends on the runtime of the input encodings and can be as small as $\polylog N$, for example, when the coefficients are sparse.
This provides a significant speedup over implementations of Newton's method on classical computers when the problem size $N$ grows faster than the reciprocal error tolerance $\tol^{-1}$. This scenario commonly arises when solving nonlinear partial differential equations (PDEs), where the number of degrees of freedom of a spatial discretization to push the error below a given tolerance typically suffers from the curse of dimensionality.

Nonlinear computations on quantum computers are difficult, mainly because all operations on quantum computers are unitary, and therefore inherently linear, transformations in the state space. We address this constraint by considering product states that encode multiple copies of the same variable, allowing us to introduce a form of nonlinearity.

Although our approach as a whole is novel, it integrates established building blocks and relates to ideas explored in previous works. The review \cite{TLLM24} provides an overview of methods for solving nonlinear PDEs on quantum computers, categorizing them into three main strategies to circumvent the restriction to unitary operations: methods using mathematical linearization, see also \cite{JL24} and the review \cite{JLY23a}, hybrid classical-quantum algorithms, such as \cite{LJM+20}, which employ Variational Quantum Computing to approximate nonlinear states, and methods that leverage the interaction between multiple copies of the solution state, like \cite{LDG+20}. Our algorithm aligns most closely with the latter category.
Algorithms outside the context of PDEs have also explored the approach of creating such tensor states to introduce nonlinearity. Examples include the general framework in~\cite{HCSS23}, a solver for nonlinear ordinary differential equations based on the fixed-point iteration~\cite{LO08}, and methods for optimizing multivariate polynomials \cite{RSW+19, GLW+21, GLWL21}. Most of these works conclude that creating tensor states leads to an exponential runtime, or rather a decay of information, rendering the approach too slow for practical use. 
Related optimization approaches~\cite{LWG+21, CGWL24} bypass the exponential cost by measuring certain quantities related to a decomposition of the polynomial. However, the runtime of these methods then scales linearly with the size of this decomposition, which, in the worst case, does not maintain an advantage over classical algorithms. Still, the tensor state approach stands out because, unlike the others, it can be applied to a broad range of problem classes and offers provably correct methods.
Finally, we want to mention a fourth group of methods~\cite{GMF21,RR23}, which are conceptually similar to the tensor state approach. However, they introduce nonlinearity by multiplying by a diagonal matrix containing the entries of the input vector. This enables Quantum Singular Value Transformation (QSVT)~\cite{GSLW19} to be used to apply polynomials element-wise to the input. As such, the technique is applicable in the special case of decoupled univariate polynomials, for which it reduces memory requirements and may be faster; see Example~\ref{ex:elementwise-poly} for details. The work~\cite{NW24} uses a construction similar to the QSVT-based one, but using rank-$1$ matrices instead of diagonal ones.

The distinguishing feature of our approach is our use of amplitude amplification. Similar to the Variable Time Amplitude Amplification (VTAA) \cite{Amb10,CGJ19,LS24}, this results in a more efficient algorithm when performing multiple steps of a Newton or fixed-point iteration. The runtime then grows only exponentially with the number of steps, as is the case with, for example,~\cite{RSW+19,LWG+21}, albeit at the cost of spherical constraints. Crucially, we relate this exponential complexity to the number of steps required to reach a prescribed tolerance, proving our approach is feasible. Due to the second-order (super-exponential) convergence of Newton's method, this results in an almost optimal overall runtime bound, an idea only hinted at in~\cite{RSW+19}. The  amplification scheme and resulting analysis may extend to any of the named tensor state or QSVT-based methods.
Additionally, we improve upon existing methods by introducing a simpler approach for computing the Jacobian of the target function, as well as its inverse. This simplification enables the development of an algorithm that has a working quantum implementation, which we have successfully tested on quantum simulators. We also demonstrate the feasibility of running the circuits on currently available commercial quantum hardware. Notably, the building blocks of our algorithm are designed to be straightforward and modular, ensuring accessibility and adaptability across a wide range of applications.

\paragraph{Notation}
To characterize the complexity of algorithms, we use the ``Big-$\mathcal{O}$-notation'', where the hidden constant is truly independent of the objects involved. Note that the hidden constant in this notation may also depend on the set of gates of the hardware used. We measure quantum algorithms in terms of the atomic gates that are being executed. This is related to the actual runtime by a multiplicative constant. The runtime of additional classical computation is negligible.
Whenever we write $\log$ we refer to the dyadic logarithm (base $2$), although in most places the particular choice of basis is not important. The operator $\polylog x$ in the $\mathcal{O}$-notation signifies $(\log x)^s$ for some constant $s \in \N$.

\section{Quantum building blocks for nonlinear computation} \label{sec:framework}
In a quantum computer, numerical information is best stored in the amplitudes of quantum states. For example, a vector $v \in \C^{2^n}$, $n \in \N$, normalized in the Euclidean norm $|\bullet|$ is stored using $n$ qubits encoded as
\[
\ket{v} \coloneqq \sum_{j = 0}^{2^n-1} v_j \ket{j}
\]
where $\ket{j}$ is the state of the qubit register that stores the integer $j$ in binary representation. Note that through this \textit{amplitude encoding} we store a normalized vector of dimension $2^n$ in only $n$ qubits, an exponential reduction in memory requirements compared to a classical computer. This encoding offers the potential for significant runtime improvements, but it also introduces challenges in preparing and limitations in manipulating these encoded quantum states. 

Since the effect of a quantum algorithm on an initial input state $\ket{\phi} \in \C^{2^n}$ can always be described by a unitary matrix $U \in \C^{2^n \times 2^n}$, the evolution of the quantum state is governed by linear transformations such that the resulting state is given by $U\ket{\phi}$. This inherent linearity is a fundamental property of quantum mechanics, ensuring that quantum operations preserve the overall probability of quantum states. However, it also means that directly implementing nonlinear transformations requires workarounds or approximations. A naive way to allow such nonlinear operations is to store numerical information in a binary format, as in a classical computer, rather than in the amplitudes of quantum states, but this approach would negate any advantages in terms of runtime and memory efficiency that arise from using the high-dimensional state space of a quantum computer. Thus, achieving nonlinearity while retaining the benefits of quantum computing requires more sophisticated strategies.

The challenge also arises when implementing linear but nonunitary operations. The concept of block encodings provides a solution by embedding these operations in a higher-dimensional unitary framework. This approach not only allows for the handling of linear transformations beyond unitary matrices, but also lays the foundation for nonlinear computations.

\subsection{Block encoding of nonlinear operations}
A fundamental nonlinear operation is the element-wise square of a quantum state represented by a normalized vector $v \in \C^{2^n}$, which we will denote by $v^2$ for simplicity. This state cannot be realized from an initial state of $\ket{v}$ through a quantum gate that represents a unitary transformation, but it can still be prepared on a quantum computer using the circuit in Figure~\ref{fig:circuit-square}. 
\begin{figure}
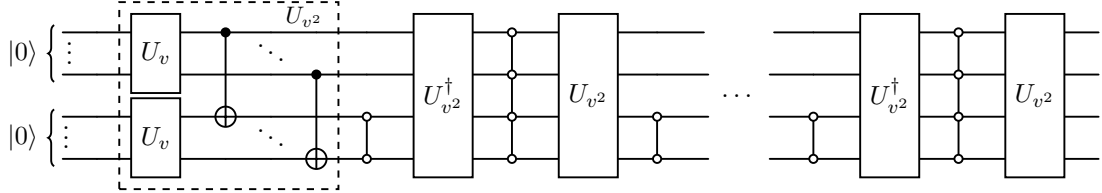

    \centering
    \includestandalone{circuit_square}
    \caption{Circuit to prepare and amplify a state $\ket{v^2}$ representing the element-wise square of a vector $v$, where $U_v$ is a circuit preparing the state $\ket{v}$.}
    \label{fig:circuit-square}
\end{figure}
Specifically, given \textit{two} copies of $\ket{v}$ in separate qubit registers, the desired quantities $v_j^2$ are implicitly encoded in the total state given by their tensor product 
\[\ket{v}\ket{v} = \sum_{j = 0}^{2^n-1} \sum_{k = 0}^{2^n-1} v_jv_k\ket{j}\ket{k} = \sum_{j = 0}^{2^n-1} v_j^2 \ket{j}\ket{j} + \sum_{j = 0}^{2^n-1} \sum_{\substack{k = 0 \\ k \neq j}}^{2^n-1} v_jv_k\ket{j}\ket{k}.\] 
The terms $v_j^2$ appear in the components where both registers are in the same state $\ket{j}$. Using quantum operations that extract only the parts of this combined state where the two registers are equal, we can effectively compute the element-wise square of $v$ on a quantum computer. This extraction can be realized by first computing the element-wise \emph{exclusive or} ($\oplus$) on the two registers,
in the quantum context also called $\operatorname{CNOT}$
\[
\operatorname{CNOT}^{\otimes n}\ket{v}\ket{v} = \sum_{j = 0}^{2^n-1} \sum_{k = 0}^{2^n-1} v_jv_k\ket{j}\ket{j \oplus k} = \ket{v^2}\ket{0} + \ket{\mathrm{bad}}.
\]
This corresponds to the subcircuit labeled $U_{v^2}$ in Figure~\ref{fig:circuit-square}, which seems to have been discovered independently multiple times, e.g.~in~\cite{LJM+20, RNF+23, GYC+24, Zyl24}. The second summand $\ket{\mathrm{bad}}$ is some state orthogonal to the subspace $\C^{2^n} \otimes \ket{0}$ and is not of any further interest. By applying a \textit{projective measurement} on the second register, the target information encoded in the first register can be extracted. Specifically, upon measuring the second register, the state collapses into a new state. With probability $|v^2|^2$, the result $0$ is measured and the new state is
\[\tfrac{1}{|v^2|}\ket{v^2}\ket{0}.\]
We only consider these runs measuring $0$ in the second register. By estimating the probability of this outcome, we can also recover the missing information about the norm $|v^2|$. It is important to note that this procedure requires two separate instances of $\ket{v}$. Due to the \textit{no cloning theorem}, these instances cannot be obtained by copying a single $\ket{v}$, and since imperfect cloning \cite{BH96} cannot be trivially applied in this context, the process that prepares the state $\ket{v}$ must be repeated independently to generate each instance.

In practice, the projective measurements used in this procedure are often not only difficult to implement, since they may require measurements in the middle of the circuit, but also somewhat inefficient, since the average time to extract the desired part of a state scales quadratically with the amplitude of that state. We can avoid the need for projective measurements by rephrasing the above approach in the language of block encodings \cite{CGJ19, CKS17, GSLW19}. Block encodings allow us to represent nonunitary linear operations by normalizing them and embedding them in higher-dimensional unitary transformations. The embedding of an $N$-dimensional space in a larger space of dimension $2^m\geq N$ is facilitated through the so-called $\CNOT{\Pi}$ gates. 
\begin{definition}[Projection as $\CNOT{\Pi}$ gate] \label{d:Cnotproj}
Let $N, m \in \N$ with $N \le 2^m$ and $\Pi \in \C^{N \times 2^m}$ be a matrix with orthonormal rows, representing a linear subspace projection. We call a gate of $m + 1$ qubits a $\CNOT{\Pi}$ gate if it flips the last bit if and only if the first $m$ bits represent a vector in the range of $\Pi^\dag$, i.e.~it performs the unitary operation 
\[ \Pi^\dagger\Pi \otimes \mathrm{X} + (\Id_N - \Pi^\dagger\Pi) \otimes \Id_2.\]
Here, $\Pi^\dagger$ denotes the conjugate transpose of $\Pi$, as is common in quantum mechanics, and $\mathrm{X}$ denotes the Pauli matrix $\mathrm{X} = \ket{0}\!\bra{1} + \ket{1}\!\bra{0}$.
\end{definition}
We will now denote by $\Pi$ not only the matrix, but also the implementation as a $\CNOT{\Pi}$ gate.
The following definition of block encoding is taken from \cite{DP25} and inspired by \cite{GSLW19}.
\begin{definition}[Block encoding of nonunitary operations \cite{DP25}] \label{d:block}
Let $\tol \ge 0$ and $m,N_1,N_2\in \N$ such that $M=2^m\geq \max\{N_1,N_2\}$, and consider a matrix $S \in \C^{N_2 \times N_1}$. We call the tuple $(U_S, \Pi_1, \Pi_2, \gamma)$ consisting of
\begin{itemize}
  \setlength\itemsep{0.2em}
    \item a quantum algorithm implementing a unitary matrix $U_S \in \C^{M \times M}$,
    \item subspace projections $\Pi_j \in \C^{N_j \times M}$, $j=1,2$ in the sense of Definition \ref{d:Cnotproj}, and
    \item the \emph{normalization factor} $\gamma \ge |S|$
\end{itemize}
a \emph{block encoding of \(S\) up to relative tolerance $\tol$} if
\[
|\gamma \Pi_2 U_S \Pi_1^\dag - S| \le \tol|S|.
\]
For convenience, we sometimes write $U_S$ to refer to the complete tuple and $\gamma(U_S)$ to refer to the normalization factor. If the tolerance $\tol$ is omitted, we assume that the block encoding is exact ($\tol = 0$).
Furthermore, the \textit{information efficiency} $\ie(U_S)$ of a block encoding is defined as
\begin{equation} \label{eq:information-efficiency}
    \ie(U_S) \coloneqq |\Pi_2 U_S\Pi_1^\dag| = |S|/\gamma \in [0, 1],
\end{equation}
where the second equality only holds if the block encoding is exact.
\end{definition}
Information efficiency, previously introduced in reciprocal form as ``subnormalization'' in \cite{DP25}, correlates with the actual cost of using a block encoding. Specifically, applying a block encoding to a state can dilute the state's information into parts of the higher-dimensional space \(\mathbb{C}^M\) that are outside the image of \(\Pi_2^\dag\), making them irrelevant to the computation. The probability of obtaining a relevant output from the produced state by naively sampling decreases by a factor of \(\ie(U)^2\) after applying the block encoding. As a result, the overall runtime increases by a factor of \(\ie(U)\), reflecting the additional effort required to obtain meaningful results from the encoded operation.

Block encodings can not only be used to realize nonunitary linear operations in quantum computers but also form the basis for nonlinear operations, as the following example shows.
\begin{example}[Block encoding of element-wise multiplication] \label{ex:mul}
The element-wise multiplication of two vectors can be understood as a linear operator $\odot \colon \C^{2^n} \otimes \C^{2^n} \to \C^{2^n}$ and the above observation shows that a block encoding of $\odot$ can be constructed with \[
U_\odot = \operatorname{CNOT}^{\otimes n},\qquad
\Pi_1 = \operatorname{Id}_{2^{2n}},\qquad
\Pi_2 = \operatorname{Id}_{2^n} \otimes (\proj{0})^{\otimes n},
\qquad\text{and}\quad\gamma = 1.
\]
The $\CNOT{\Pi}$ gates for the projections $\Pi_1$ and $\Pi_2$ can be implemented as $\Id_{2^{2n}} \otimes \mathrm{X}$ and $\Id_{2^n} \otimes \operatorname{CNOT}_n$ respectively, where $\operatorname{CNOT}_n$ denotes the $n$-qubit multi-controlled NOT-gate with control state $\ket{0}^{\otimes n}$.
\end{example}

The notion of block encodings is trivially extended also to vectors, by considering them as column matrices. For such encodings, we always assume that the projection $\Pi_1 \in \C^{1 \times M}$ is given by $\bra{0}$ and typically write $\Pi = \Pi_2$.
We also recall from \cite[Proposition~3.3]{DP25} that block encodings allow for a variety of operations, including addition, matrix-matrix and matrix-vector multiplication, and tensor products.
\begin{example}[Block encoding of element-wise powers of a vector] \label{ex:vector-power}
Let $U_v$ be a block encoding for some normalized vector $v \in \C^{2^n}$. Then we can first use the tensor product operations to obtain an encoding of $v \otimes v$. Applying the matrix-vector product operation to the result of this and $U_\odot$ leads to a block encoding of $v^2$, which corresponds to the subcircuit labeled $U_{v^2}$ in Figure~\ref{fig:circuit-square}.
By repeating the procedure, any element-wise power $v^{s}$, $s \in \N$, of a vector $v$ can be constructed. However, the no-cloning principle requires that the state $v$ be prepared at least $s$ times with the subroutine $U_v$ when using this technique. This is in contrast to classical computers where only $\log_2(s)$ calculations of the intermediate results $v^{2^j}$, $j=1,\ldots,\log_2(s)$ suffice.
\end{example}

Note that in the above case, the composition of the two block encodings $U_{v \otimes v}$ and $U_\odot$ works only if the projection $\Pi_2$ of $U_v$ is also given by $\Pi_2 = \Id_{2^{2n}}$, so that the embedding of $\C^{2^{2n}}$ is the same for both encodings. Restrictions to projections like this one are, in many cases, more a technical detail than a theoretical restriction, and can be resolved by finding a circuit that implements a unitary $V$ such that the equality $\Pi_2 V = \Id_{2^{2n}}$ holds. The unitary can then be appended to $U_v$, turning it into a block encoding with the desired projection. We make this slightly more rigorous with the following definition:

\begin{definition}[Equivalence of projections]
    Let $N, m_1, m_2 \in \N$. We say that two projections $\Pi_1 \in \C^{N \times 2^{m_1}}$ and $\Pi_2 \in \C^{N \times 2^{m_2}}$ are \emph{equivalent}, if we can find a circuit that implements a unitary $V \in \C^{m_1 \times m_1}$ such that for $m \coloneqq \min\{m_1, m_2\}$ we have \[(\Pi_1 V)|_{\C^{2^m}} = \Pi_2|_{\C^{2^m}}.\]
\end{definition}

One can check that this indeed defines an equivalence relation.
Clearly, two projections are equivalent if they are equal. Furthermore, they are equivalent if $\Pi_j = \Pi \otimes \ket{\phi_j}$, $j = 1,2$ for some known constant states $\ket{\phi_j}$ that can be prepared efficiently and with fixed projection $\Pi$.
In the following, we will assume that all projections with the same image are equivalent without further mention.

We conclude that given a block encoding of a vector, it is indeed possible to derive block encodings for nonlinear functions applied to this vector. However, this means that the original block encoding of the vector must be used several times, possibly very often. As we shall see, the advantageous dependence of quantum algorithms on the dimension of the space $\C^{2^n}$ can still give a significant improvement over classical computations. An important hurdle to achieve this is to prevent the loss of information efficiency associated with block encodings and their naive composition.

\subsection{Amplification of block encodings} \label{ssec:amplification}
In Example~\ref{ex:vector-power} discussing the element-wise $s$-th power of a vector $v$, the tensor product leads to an encoding $U_{v^{s}}$ with information efficiency
\[
\ie(U_{v^{s}}) = \ie(U_v)^{s}.
\]
Unless the block encoding of $v$ had a perfect information efficiency $\ie(U_v)=1$, this implies that the information efficiency of the encoding of the $s$-th power degrades exponentially with power $s$. This exponential loss of information efficiency can be mitigated by a process called \textit{amplitude amplification}. The following lemma is a classical result, closely related to the Grover search~\cite{Gro98}, and has been restated multiple times, such as in \cite{BH97,Amb10}, and in \cite[Theorem~28]{GSLW19} in the context of block encodings.
\begin{lemma}[Amplitude amplification] \label{lem:amplitude-amplification}
    Let $(U, \Pi, \bra{0}, \gamma)$ be a block encoding of a vector $v$ with information efficiency $\eta(U)$. Then for any odd $k \in \N$ there exists a block encoding $(\hat U, \Pi, \bra{0}, \hat \gamma)$ that also encodes $v$, uses $U$ exactly $k$ times and has information efficiency
\[ \ie(\hat U) = \big|\sin\big(k\sin^{-1}(\ie(U))\big)\big|. \]
\end{lemma}
For a proof, we refer to \cite[Theorem~28]{GSLW19}, where the statement is restricted to the case where $\ie(U) = \sin(\tfrac{\pi}{2k})$, which implies $\ie(\hat U) = 1$. However, it is easily verified that the proof remains valid without this restriction and works for any value of $\ie(U)$.

By properly choosing the number of iterations $k$ in this lemma, the information efficiency of a block encoding $U$ of a vector can be constructively improved. This is done by alternately applying $U$ and $U^\dag$ alternating with the relative phase gates $2\proj{0} - \Id$ and $2\Pi^\dag\Pi - \Id$, which leads to the implementation of the unitary $\hat U$. This process is illustrated in the amplified block encoding of the element-wise square from Example~\ref{ex:mul} in Figure~\ref{fig:circuit-square}. To apply to the computation of $v^s$, one could hope to amplify $U_v$ to perfect information efficiency, which would incur a multiplicative factor proportional to $\ie(U_v)^{-1}$ in the runtime, but counteract the exponential decay of information efficiency. In practice, achieving perfect information efficiency is not always possible.
This is because to properly choose the normalization $\hat \gamma$ of the amplified encoding, the information efficiency~$\ie(U)$ must be explicitly known. If $\eta(U)$ is not known, it can be estimated using the \textit{amplitude estimation} \cite{SUR+20,RF23}, which leads to the following construction.
\begin{theorem}[Normalization of block encodings] \label{thm:normalizing}
There exists a constant $\Cie \ge \tfrac14$ and an algorithm with the following behavior. The algorithm takes as input parameters a tolerance $0 < \tol \le 1$, a failure probability $0 < \delta < 1$, and a block encoding $U$ of a vector $x \in \C^N$ (for some $N \in \N$) up to a nonzero information efficiency $\ie(U) > 0$. The algorithm is successful with a probability of at least $1 - \delta$. When successful, the algorithm produces a new block encoding $\hat U$ of $x$ up to relative error tolerance $\tol$ that achieves a uniformly bounded information efficiency \[ \ie(\hat U) \ge \Cie.\] 
Given the number of gates $t(U)$ in the input block encoding $U$, the number of gates executed by the algorithm is bounded by
\[\mathcal{O}\big(\tol^{-1}\ie(U)^{-1} \log \delta^{-1} t(U)\big)\]
while the number of gates in its output block encoding $\hat U$ is bounded by
\[\mathcal{O}\big(\ie(U)^{-1} t(U)\big).\]
\end{theorem}
\begin{proof}
The proof uses the amplitude amplification of Lemma~\ref{lem:amplitude-amplification}. The proper choice of the odd number of amplification steps $k$ requires estimation of the information efficiency $\ie(U)$. Using relative-error amplitude estimation as proposed in \cite{AR20}, we get the reliable estimate~$\sigma$ with error bound 
\begin{equation} \label{eq:normalizing:accuracy}
    |\sigma - \ie(U)| \le \min\big\{\tfrac{1}{2\pi+1} \tol |x|/ \gamma, \tfrac1{4\pi+1} \ie(U)\big\} = \min\big\{\tfrac{1}{2\pi+1} \tol,\tfrac1{4\pi+1} \big\}\ie(U),
\end{equation}
where we used the definition~\eqref{eq:information-efficiency} of $\ie(U)$.
To push the failure probability of this approach below $\delta$, $U$ is used at most
\[
\mathcal{O}\big(\tol^{-1}\ie(U)^{-1}\log \delta^{-1}\big)
\]
times.

Given the estimate $\sigma$, define
\begin{equation} \label{eq:amplification-steps}
k(\sigma) \coloneqq 2\bigg\lfloor \frac{\pi}{4\sin^{-1}(\sigma)} + \frac{1}{2}\bigg\rfloor - 1
\end{equation}
to be the largest odd $k$ satisfying $k\sin^{-1}(\sigma) \le \tfrac\pi2$. Applying Lemma~\ref{lem:amplitude-amplification} with this choice, we obtain an amplified block encoding $\hat U$ that also encodes $x$, but with information efficiency
\[\ie(\hat U) = |\sin(k\sin^{-1}(\ie(U)))|.\]
We indicate all variables related to the amplified encoding by a hat.
The number of steps satisfies 
\begin{equation}\label{e:numbersteps}
k(\sigma)\sigma \le \frac\pi2.
\end{equation}
In particular, the choice of $k(\sigma)$ results in $k(\sigma) = 1$ for $\sigma > \tfrac\pi6$, i.e.~$\hat U = U$ in this case. As such, in the following analysis we only need to consider the case $\sigma \le \tfrac\pi6$. By~\eqref{eq:normalizing:accuracy} we then know that 
\begin{equation}\label{e:etasigma}\ie(U) \le \frac{4\pi+1}{4\pi}\sigma \le \frac{4\pi+1}{24} < \frac34.
\end{equation}
By the choice of $k(\sigma)$ we know $(k(\sigma)+2)\sin^{-1}(\sigma) > \tfrac\pi2$. It follows $\sin^{-1}(\sigma) > \frac{\pi}{2(k(\sigma)+2)}$, and thus
\begin{equation}\label{e:hatsigma}
\hat\sigma \coloneqq |\sin(k(\sigma)\sin^{-1}(\sigma))| > \sin\biggl(\frac{\pi k(\sigma)}{2(k(\sigma)+2)}\biggr)\geq \sin\biggl(\frac\pi6\biggr)=\frac12,
\end{equation}
since clearly $\kappa(\sigma) \ge 1$.
Note that in the limit $\varepsilon \to 0$ we have $\hat \sigma = \ie(\hat U)$. To bound the error of this approximation, define the function $g\colon x \mapsto \sin(k(\sigma)\sin^{-1}(x))$ and observe that 
$$|g'(x)|=k(\sigma)|\cos(k(\sigma)\sin^{-1}(x))|(1 - x^2)^{-1/2}\leq 2k(\sigma)$$ for $ x\in [0, \tfrac34]\supset\{\sigma,\ie(U)\}$. Applying  the mean value theorem to $g$ and using  \eqref{eq:normalizing:accuracy}, the first inequality of \eqref{e:etasigma} and \eqref{e:numbersteps} gives
\begin{align}
    |\hat \sigma - \ie(\hat U)|
    &\leq |g(\sigma)-g(\ie(U))| \leq  2k(\sigma) |\sigma - \ie(U)| \stackrel{\eqref{e:numbersteps}}{\le} \pi\sigma^{-1}|\sigma - \ie(U)| \label{eq:normalizing:chebyshev}\\
    &\stackrel{\eqref{eq:normalizing:accuracy}}{\le} \pi\sigma^{-1}\frac{1}{4\pi+1}\ie(U) \stackrel{\eqref{e:etasigma}}{\leq} \pi\sigma^{-1}\frac{1}{4\pi}\sigma = \frac{1}{4}.\nonumber
\end{align} 
This shows $$\ie(\hat U) > \hat\sigma-|\hat\sigma-\ie{\hat U}|>\frac12 - \frac14 = \frac14.$$

It remains to define the normalization $\hat \gamma$ of $\hat U$, which is not provided by Lemma~\ref{lem:amplitude-amplification}, and bound the resulting error. In case the information efficiency $\ie(\hat U)$ is known exactly, the correct value can be computed as $\hat \gamma = \gamma \ie(U)/\ie(\hat U)$, since
\[
    \gamma\ie(U) = |\gamma\Pi U \ket{0}| = |v| = |\hat\gamma\Pi\hat U\ket{0}| = \hat\gamma \ie(\hat U).
\]
The normalization of $\hat U$ is hence chosen as \(\hat\gamma_\sigma \coloneqq \gamma\sigma/\hat\sigma\) and the output encoding is given by $(\hat U, \Pi, \bra{0}, \hat \gamma_\sigma)$. This encodes $x$ with absolute error 
\begin{equation}\label{eq:absUhat}
    |\hat \gamma_\sigma \Pi \hat U\ket{0} - x| = |(\hat \gamma_\sigma - \hat \gamma)\Pi \hat U\ket{0}| = |\hat \gamma_\sigma - \hat \gamma| \ie(\hat U).
\end{equation}
The combination of \eqref{eq:normalizing:chebyshev}, \eqref{e:hatsigma} and \eqref{eq:normalizing:accuracy} yields
\begin{align*}\label{eq:normalizing:tolerance}
|\hat \gamma_\sigma - \hat \gamma|\ie(\hat U) = \gamma\left| \frac{\sigma \ie(\hat U)}{\hat\sigma} - \ie(U)\right| &\le \gamma\left( |\sigma - \ie(U)| + \frac{\sigma}{\hat \sigma}|\ie(\hat U) - \hat \sigma|\right) \notag \\
&\le \gamma |\sigma - \ie(U)|(1 + \pi \hat\sigma^{-1}) \le \frac{2\pi+1}{2\pi+1}\tol|x| = \tol|x|.
\end{align*}
Together with \eqref{eq:absUhat}, this readily yields the assertion.
\end{proof}

\begin{remark}[Practical improvements of the normalization algorithm]
While relative amplitude estimation is necessary for the theoretical correctness, in practice -- especially for noisy quantum computers -- one would rather use a Monte Carlo estimator with a predefined precision, or faster absolute-error amplitude estimation like \cite{SUR+20} if the hardware allows it.
Choosing $k$ by~\eqref{eq:amplification-steps} is also not necessarily optimal.
VTAA, for example, chooses a lower $k$ to increase the efficiency of the amplification. For nonlinear computations, a speed-up can be obtained by using higher accuracy in~\eqref{eq:normalizing:accuracy}, intentionally decreasing the information efficiency $\ie(U)$ of the original encoding, and taking larger $k$ to get $\Cie$ to be almost $1$. The latter approach is chosen for the numerical experiments in Section~\ref{sec:numerics}.
\end{remark}
In both this normalization algorithm and in Example~\ref{ex:vector-power}, we observe that algorithms based on block encodings -- particularly in the nonlinear context -- generally consist of two distinct phases: a ``construction phase'' where the final block encoding $\hat U$ is assembled from some input encoding~$U$ (which may also involve the execution of $U$), and an ``execution phase'' where $\hat U$ is used for computations. We formalize this observation by the following definition of \textit{amplified encodings} of vectors and vector-valued functions, which we complement by \textit{constructive encodings} of matrices and matrix-valued functions.

\begin{definition}[Constructive encoding and amplified encoding] \label{d:implementation}
Let $N_1, N_2, N_3 \in \N$, and $\Cie\geq\frac{1}{4}$ be the constant of Theorem~\ref{thm:normalizing}. An algorithm $\A$ is called a \emph{constructive encoding} in the following cases:
\begin{enumerate}
    \item[(i)] $\A$ is a constructive encoding of a matrix $S \in \C^{N_2 \times N_1}$ if, given a failure probability $0 < \delta < 1$ and a relative error tolerance $0 < \tol < 1$, it produces a block encoding $\A(\tol)$ of $S$ with said error tolerance.
    \item[(ii)] $\A$ is a constructive encoding of a vector $v \in \C^{N_2}$, if it is a constructive encoding of $v$ as a column matrix and the projection $\Pi_1$ of the produced block encoding $\A(\tol)$ is given by $\Pi_1 = \bra{0}$.
    \item[(iii)] $\A$ is a constructive encoding of a function $f \colon \C^{N_3} \to \C^{N_2 \times N_1}$ or $f \colon \C^{N_3} \to \C^{N_2}$ if, given a failure probability $0 < \delta < 1$, an error tolerance $0 < \tol < 1$, and an input block encoding~$U_x$ of a vector~$x \in \C^{N_3}$ with $\ie(U_x) \ge \Cie$, it produces a block encoding $\A(\tol, U_x)$ of $f(x)$ with said error tolerance.
\end{enumerate}
If $\A$ always returns a block encoding with no error ($\tol = 0$), we call it \emph{exact}. Moreover, a constructive encoding $\A$ of a vector or vector-valued function is called \emph{amplified encoding}, if every output block encoding $\A(*)$ satisfies
\[
\ie(\A(*)) \ge \Cie.
\]
\end{definition}

The concept of constructive encoding has broad applications in quantum computing, beyond the treatment of nonlinearity. Many widely used quantum algorithms are naturally divided into a construction phase and an execution phase. For example, in Quantum Singular Value Transformation (QSVT)\cite{GSLW19}, there is an initial computational cost to determine the so-called \emph{phase angles}. A more direct analogy is Variable Time Amplitude Amplification (VTAA) \cite{Amb10,CGJ19} and the recent Tunable VTAA~\cite{LS24}, where intermediate-state amplitudes are estimated and amplified iteratively.
However, estimating the runtime of constructive encodings can be a bit tricky. The runtime of the construction algorithm is often less critical because it only runs once, but it typically exceeds the runtime of the constructed encoding by a factor of $\tol^{-1}$. This is balanced by the fact that the block encoding produces quantum data, where each measurement also incurs this multiplicative factor. We formalize this relationship in the following definition.
\begin{definition}[Complexity of constructive encodings] \label{d:complexity}
For a constructive encoding $\A$ of a function, we define its \emph{complexity} as
\[ \qc(\A, \tol, U_x) \coloneqq \frac{t(\A(\tol, U_x))}{t(U_x)}
, \]
where $t(\A(\tol, U_x))$ and $t(U_x)$ are the number of atomic gates of the respective block encodings. Often, the complexity $\qc(\A, \tol, U_x)$ is independent of $\tol$ and the specific encoding $U_x$ of $x$, in which case we simply denote it by $\qc(\A, x)$. Similarly, if $\A$ encodes a vector or matrix rather than a function, we set $t(U_x) = 1$ and write $\qc(\A)$.

We say that $\A$ is \emph{efficient} if the number of gates $t(\A, \tol, \delta, U_x)$ executed by $\A$ satisfies
\[t(\A, \tol, \delta, U_x) \in \mathcal{O}\big(t(\A(\tol, U_x))\tol^{-1}\log\delta^{-1}\big) \]
with a constant independent of $x$, $\tol$, or $\delta$.
If $\A$ is exact, it is trivially efficient.
\end{definition}

This notion of complexity does not track the memory requirement of the block encodings, which will usually scale like $\log N$ where $N$ is the maximum dimension of the considered vector spaces.
Although the definition is somewhat technical, it greatly simplifies the following analysis. The statement of Theorem~\ref{thm:normalizing} can, for example, be rephrased as follows.
\begin{corollary} \label{cor:amplify-constructive}
Let $\A$ be an efficient constructive encoding of a vector or vector-valued function, and let $\ie_{\min} > 0$ be such that the block encodings $\A(\tol)$ satisfy $\ie(\A(\tol)) \ge \ie_{\min}$ for all $0 < \tol \le 1$. Then we can find an efficient amplified encoding $\A'$ of the same vector or function. The new encoding $\A'$ has complexity
\[
\qc(\A', \tol) \in \mathcal{O}(\ie_{\min}^{-1}\qc(\A, \tfrac12\tol)).
\]
\end{corollary}

For matrix-valued encodings, a similar approach cannot be directly applied without additional assumptions, as measuring the operator norm of a matrix is more complex. Theoretically, knowing the matrix norm, one could use uniform amplitude amplification such as~\cite[Theorem~30]{GSLW19} to achieve a similar effect. However, we only assume access to constructive encodings for matrices and specify the information efficiency of the returned block encoding to accurately assess the execution cost.

Amplifying an encoding is particularly useful because it allows the output to serve as the input to the encoding $\F$ of a function, as specified in Definition~\ref{d:implementation}(iii). This composition ensures that the complexity $\qc(\F, \hat U)$ becomes independent of $\ie(U)$. Consider, for example, the case where $\F$ is a constructive encoding of the function $v \mapsto v^s$ previously discussed in Example~\ref{ex:vector-power}. If the unamplified encoding $U$ is used as input to this encoding, the final block encoding will have an information efficiency of $\ie(U)^s$ and a complexity proportional to $s$. In contrast, using the amplified $\hat U$ produces an information efficiency of $\Cie^s$ and a complexity proportional to $s\ie(U)^{-1}$. Multiplying the complexity by the inverse information efficiency reveals that in the first case, the measurement runtime depends on $\ie(U)^s$, while in the second case it scales only linearly with $\ie(U)$.
More complex mappings are also possible, including general multivariate polynomials and matrix inversion, which we cover in the following.

\subsection{Operations on amplified encodings}
We demonstrate that basic arithmetic operations, as well as fundamental matrix and function operations, can be efficiently performed on constructive encodings; see Table~\ref{tab:framework} for an overview. They are implemented by lifting operations on block encodings, as described in \cite[Proposition~3.3]{DP25}. Together, amplified encodings and these operations create a versatile and abstract framework for numerical computation on quantum computers.

\begin{table}
    \centering
    \begin{tabular}{l|c|c}
        operation & complexity & \\
        \hline
        sums \& products & $\qc(\A) + \qc(\B)$ (unamplified) & Lemma~\ref{lem:ops} \\
        multivariate polynomials & $K(|A_0| + \dots + |A_K||x|^K) / |f(x)|$ & Theorem~\ref{thm:block-encoding-poly} \\
        inversion / division & $\kappa\log \tol^{-1}(\qc(\A) + \qc(\B))$ & Lemma~\ref{lem:invert} \\
        function application & $\qc(\F, x)\qc(\A)$ & Lemma~\ref{lem:composition} \\
        measurement & $\qc(\A)\tol^{-1}\log \delta^{-1}$ & Lemma~\ref{lem:measure}  \\
    \end{tabular}
    \caption{Overview of operations introduced for constructive and amplified encodings.}
    \label{tab:framework}
\end{table}

\begin{lemma}[Arithmetic operations on amplified encodings] \label{lem:ops}
    Let \(A\) and \(B\) be matrices, vectors or functions of appropriate dimension and assume that we have access to exact constructive encodings $\A$ and $\B$ of either value. We can then find exact constructive encodings for
    \[A \otimes B, \qquad AB, \quad \text{and} \quad A + B.\]
    The complexity is given by the sum $\qc(\A) + \qc(\B)$ and the normalization factors of the output encodings are given respectively by
    \[
    \gamma_{A \otimes B} = \gamma_{AB} = \gamma_A\gamma_B
    \qquad\text{and}\qquad
    \gamma_{A + B} = \gamma_A + \gamma_B.
    \]
\end{lemma}

The normalization factors will allow us to determine the information efficiency by its definition~\eqref{eq:information-efficiency}. Although the lemma could be generalized to cases where $\A$ and $\B$ are not exact, deriving the corresponding error bounds, particularly for addition, is rather technical at this abstract level. However, in Section~\ref{sec:schemes}, we will integrate approximate operations within the specific context of nonlinear solvers.

The constructive encodings given by these operations can be used to create amplified encodings of general (possibly nonlinear) mappings. As an example, we show how to construct amplified encodings for multivariate polynomials and their Jacobian given block encodings of the coefficients. Polynomials are particularly interesting because they are easy to implement and can approximate any continuous function.
For this, let $N \in \N$ and $f \colon \C^N \to \C^N$ be a multivariate polynomial
\begin{equation*} \label{eq:polynomial-approximation}
f(x) = A_0 + A_1x + A_2(x \otimes x) + \dots + A_Kx^{\otimes K}
\end{equation*}
where $A_k \in \C^{N \times kN}$, $k = 0, \dots, K$ are symmetric in their parameters. Given block encodings of these coefficients, Lemma~\ref{lem:ops} allows us to find an amplified encoding for $f$.

\begin{theorem}[Amplified encoding of polynomial and its Jacobian] \label{thm:block-encoding-poly}
    Let $N, K \in \N$. Let $A_k \in \C^{N \times kN}$ for $k = 0, \dots, K$ and $x \in \C^N$. We assume access to the block encodings $U_k$ of $A_k$. Let $t_{\max} > 0$ be an upper bound for the number of atomic gates and $\ie_{\min}$ a lower bound for the information efficiency of these block encodings. We can then find an efficient amplified encoding of
    \[
    f(x) = A_0 + A_1x + \dots + A_Kx^{\otimes K} \in \C^N
    \]
    with complexity
    \[
    \mathcal{O}\left(
    \frac{K t_{\max}}{\ie_{\min}|f(x)|} \bigg[\sum_{k=0}^K |A_k|\,|x|^k\bigg]\right).
    \] 
    Additionally, we can find an exact constructive encoding of 
    $$\Diff f(x) = \sum_{k=1}^K k A_k(x^{\otimes (k-1)} \otimes \Id) \in \C^{N \times N}$$
    with information efficiency 
    \[
    \ie \ge 
    \ie_{\min} \Cie^{K-1} |\Diff f(x)|\bigg[\sum_{k=1}^K k\,|A_k|\,|x|^{k-1}\bigg]^{-1}
    \]
    and complexity $(K-1)t_{\max}$.
\end{theorem}
\begin{proof}
Using the operations of Lemma~\ref{lem:ops} naively to build a block encoding of $f(x)$ requires $K(K + 1)/2$ applications of the input encoding $U_x$. With a slight modification, this number can be reduced to $K$. Consider the matrices $B_{k,x} \in \C^{N \times kN}$, $k = 0, \dots, K$ that are iteratively constructed by
\[B_{K,x}=A_K \quad\text{and}\quad
B_{k,x} = A_k + B_{k+1,x}(x \otimes \Id^{\otimes k}) \quad\text{for}\quad k = K-1,K-2, \dots, 0.
\]
Then clearly $f(x) = B_{0,x}$ and the construction uses $x$ exactly $K$ times, making the total complexity $Kt_{\max}$. Moreover, given an input encoding $U_x$ of $x$, the normalization of the resulting block encoding is given by
\[
\sum_{k=0}^K \gamma(U_k)\gamma(U_x)^k = \sum_{k=0}^K \frac{|A_k|}{\ie(U_k)} \frac{|x|^k}{\ie(U_x)^k} \le \frac{1}{\ie_{\min}\Cie^K} \sum_{k=0}^K |A_k| |x|^k,
\]
where $\Cie$ is the constant of Theorem~\ref{thm:normalizing}. The information efficiency is this value divided by $|f(x)|$.
We then apply Corollary~\ref{cor:amplify-constructive} to prove the assertion for $f(x)$.
Following the same approach yields the assertion for $\Diff f(x)$.
\end{proof}

Here, the dimension $N$ enters through the number of gates $t_{\max}$ in the block encodings of coefficients. This dependency can potentially be as small as $t_{\max} \in \mathcal{O}(\log N)$.
To implement Newton's method, we intend to apply the inverse matrix of $\Diff f(x)$. For our analysis, we can use algorithms for inversion based on VTAA, which, as mentioned, is conceptually related to constructive encodings.

\begin{lemma}[Matrix inversion on amplified encodings] \label{lem:invert}
    Let $N_1, N_2 \in \N$, $\B$ be an efficient amplified encoding of $b \in \C^{N_2}$ and $\A$ be an exact constructive encoding of $A \in \C^{N_2 \times N_1}$ where $A$ is invertible with condition number $\kappa$. Alternatively, they may be encodings of functions that return such values. Then we can construct an efficient amplified encoding of $x \coloneqq A^{-1}b$ with complexity
    \[\mathcal{O}\big(\kappa \log \tol^{-1} \ie(\A)^{-1}(\qc(\A) + \qc(\B, \tol/\kappa))\big).\]
\end{lemma}
A proof is provided e.g.\ by~\cite[Theorem~4]{LS24}. The dependence of the complexity on $\qc(\B)$ could technically be improved from $\kappa \log \tol^{-1}$ to $|A^{-1}b|/|A^{-1}|$.
The second ingredient, needed for both the Newton's method and fixed-point iteration, is simple function evaluation.

\begin{lemma}[Application of amplified encodings] \label{lem:composition}
    Let $N \in \N$.
    Consider an exact amplified encoding $\A$ of a vector $x \in \C^{N}$ and an efficient amplified (or constructive) encoding $\F$ of a function $f$ taking a vector $\C^N$ as input. Then $\F \circ \A$ is an efficient amplified (or constructive) encoding of $f(x)$ with complexity
    \[
        \qc(\F, x) \qc(\A).
    \]
\end{lemma}

A similar statement could be made for the function composition. 
Finally, to actually extract information from these encodings, we need some kind of measurement. Using relative-error amplitude estimation as in the proof of Theorem~\ref{thm:normalizing} one can easily estimate the norm of an encoded vector.

\begin{lemma}[Measurement of amplified encodings] \label{lem:measure}
    Let $\tol > 0$, $0 < \delta < 1$, and $N \in \N$. Let $\A$ be an efficient amplified encoding of a vector $v \in \C^N$. Then the norm $|v|$ can be estimated up to relative tolerance $\tol$ and with a success probability of $1 - \delta$ using at most
    \[
    \mathcal{O}(\qc(\A) \tol^{-1}\log \delta^{-1})
    \]
    atomic gates.
\end{lemma}

Other kinds of measurement could be derived from this lemma, for example, by composing $\A$ with an encoding for a function representing a quantity of interest, or by combining multiple measurements to compute an inner product using the Hadamard test \cite{AJL06}.

The question remains, whether efficient amplified encodings of relevant matrices can be constructed. Indeed, Theorem~\ref{thm:block-encoding-poly} partially answers this question for the encoding of multivariate polynomials, but the format in which the coefficients are given -- block encodings of the coefficients -- is somewhat unusual. Therefore, we will provide three specific examples.

% \begin{example}[Amplified encoding for nonlinear PDEs]
%     Consider a nonlinear system of equations of the form
%     \[
%     f(x) = Sx + kx^2 - b = 0,
%     \]
%     with left-hand side matrix $S \in \R^{N \times N}$, right-hand side vector~$b$, and a coefficient $k > 0$ acting as a weight for the nonlinear term $x \odot x$. Such an equation would arise, for example, from finite difference discretization a the nonlinear partial differential equation of the form
%     \[
%     -\Delta u + k |u|^2 = f.
%     \]
%     In this case $S$ would be a stiffness matrix and $b$ would contain certain evaluations of the source function~$f$. It is well known how to construct a block encoding $U_S$ and $U_b$ for such $S$ and $b$ \cite{DP25}. Clearly, a block encoding of the function
%     \[
%     x^{\otimes 2} = x \otimes x \mapsto k x^2 = k(x \odot x)
%     \]
%     matches $U_\odot$ from Example~\ref{ex:mul} besides a scaling of the normalization by $k$. Thus, an efficient amplified encoding follows from the application of Theorem~\ref{thm:block-encoding-poly} with
%     \begin{align*}
%     A_0 &= -b, & A_1 &= S, & A_2 &= k \odot \\
%     U_0 &= -U_b, & U_1 &= U_S, & U_2 &= (U_\odot, \Pi_{1,\odot}, \Pi_{2,\odot}, k).
%     \end{align*}
% \end{example}

\begin{example}[Amplified encoding of element-wise polynomials] \label{ex:elementwise-poly}
      Let
  \[
    P(x) \;=\; c_{1}x + c_{2}x^{2} + \cdots + c_{K}x^{K}
  \]
  be a univariate polynomial, and consider the map
  \[
    x \;\mapsto\; \bigl(P(x_{0}),\,P(x_{1}),\,\dots,\,P(x_{N-1})\bigr)
  \]
  which applies \(P\) simultaneously to each coordinate of the input vector \(x \in \mathbb{R}^{N}\).  An efficient amplified encoding of this map can be constructed using Theorem~\ref{thm:block-encoding-poly} by setting
    \[
    A_k(x^{k\otimes}) = k \underbrace{x \odot \dots \odot x}_\text{$k$ times}.
    \]
    This can be implemented using $U_\odot$ from Example~\ref{ex:mul}. The resulting block encodings have information efficiency $\ie = 1$ and the complexity of the amplified encoding can be estimated using Theorem~\ref{thm:block-encoding-poly}, given the minimum value of $P$.
    For this specific case of element-wise polynomials, the QSVT based method \cite{RR23} offers improved performance. (Recall that in our framework, QSVT is typically only used for amplification, but is unrelated to the nonlinearity.) The first step is to encode $x$ as the diagonal matrix
    \[ D = \operatorname{diag}(x_0, \dots, x_{N-1}).\]
 Its construction is equivalent to the partial application $\odot(\Id \otimes x)$ of the block encoding~$U_\odot$. Then, using QSVT, polynomials of $D$ can be implemented and  multiplied with $x$ in a process known as importance sampling. Altogether, the polynomial is computed as
    \[P(x) = \tilde P(D) x \qquad\text{where }\tilde P(x) = P(x) / x.\]
    Through the use of QSVT, the factor $\sum_k |A_k| = \sum_k |c_k|$ is improved to $\max_{x\in[-1,1]} |P(x)/x|$, cf.~\cite[Theorem~3]{RR23}, which may be smaller for $K \ge 3$.
\end{example}

\begin{example}[Amplified encoding of self convolution] \label{ex:convolution}
To illustrate a case that cannot be easily implemented using QSVT-based methods, consider the cyclic convolution 
\[
v \ast v \coloneqq
\begin{bmatrix}
    v_0v_0 + v_1 v_{N-1} + v_2 v_{N-2} + \dots + v_{N-1}v_1 \\
    v_1v_0 + v_2 v_{N-1} + v_3 v_{N-2} + \dots + v_0v_1 \\
    \vdots \\
    v_{N-1}v_0 + v_0 v_{N-1} + v_1 v_{N-2} + \dots + v_{N-2}v_1 \\
\end{bmatrix}
\]
of a vector $v \in \C^N$ with itself where $N = 2^n$. Note that the $j$-th component of $v$ contains all products $v_kv_\ell$, where $k + \ell = j \operatorname{mod} N$.
An amplified encoding of this operation can be implemented using a quantum algorithm~$U_+$ that performs integer addition modulo $2^n$
\[
U_+ \ket{x}\ket{y} = \ket{x + y \operatorname{mod} 2^n}\ket{y}.
\]
It is easily verified that this constitutes a block encoding of the operation $v \otimes v \mapsto v \ast v$ by setting $\gamma = 2^{n/2}$, $\Pi_1 = \Id_{2^{2n}}$, and $\Pi_2 = [1/\sqrt{2} \; 1/\sqrt{2}]^{\otimes n} \otimes \Id_{2^n}$.
\end{example}

\begin{example}[Amplified encoding of sparse polynomials] \label{ex:sparse}
    It is well known that, given oracles providing the non-zero entries and sparsity pattern of a matrix, one can efficiently construct a block encoding of any sparse matrix, see e.g.~\cite[Lemma~48]{GSLW19}. This approach can be used to find efficient amplified encodings of polynomials of the form
    \[
    P(x_0, \dots, x_{N-1})_n = \sum_{k = 0}^K \sum_{j = 0}^{J-1} a_{j,k,n} x^{\alpha^{(j,k,n)}} \qquad\forall n = 0, \dots, N-1
    \]
    where $J, K \in \N$ are sparsity and degree of $P$, $a_{j,k,n} \in \C$ are coefficients, and $\alpha^{(j,k,n)} \in \{0, \dots, N-1\}^{k}$ are multi-indices. In the above formula, we use the notation
    \[
    x^\alpha \coloneqq x_{\alpha_0} x_{\alpha_1} \dots x_{\alpha_{k-1}}.
    \]
    The format of $P$ ensures that the corresponding coefficient matrices $A_k$ in  Theorem~\ref{thm:block-encoding-poly} have at most $J$ non-zero entries per column. Assuming further that the mapping $n \mapsto \alpha^{(j,k,n)}$ is injective for all $j, k$ ensures that all rows have at most $J$ non-zero entries as well. Using \cite[Lemma~48]{GSLW19}, we can therefore find block encodings of the matrices $A_k$ given three oracles (for each $k = 0, \dots, K$) which compute the maps $(j,n) \mapsto \alpha^{(j,k,n)}$, its inverse with respect to $n$, i.e.~$(j, \alpha^{(j,k,n)}) \mapsto n$, and $(j,n) \mapsto a_{j,k,n}$ respectively. This yields an amplified encoding of $P$, where the complexity can be bounded using Theorem~\ref{thm:block-encoding-poly} with $\eta_{\min} = J$. In particular, if the runtime of the three oracles scales polylogarithmically with $N$, then this will also be the case for the complexity of the amplified encoding. Such efficiency is, for example, possible for the solution of partial differential equations such as
    \[
    -\Delta u + a u^2 = f \quad\text{in }[0, 1]^d, \qquad u = 0\quad\text{on }\partial [0, 1]^d
    \]
    using a finite element or finite difference scheme~\cite{DP25}.
\end{example}

This concludes our conceptual framework for numerical computation on quantum computers that is broad and conceptually applicable to a number of domains, including partial differential equations. To demonstrate its suitability for nontrivial operations represented by iterative algorithms, we use nonlinear solvers as a proof of concept and show that they can be implemented straightforwardly using the introduced operations.

\section{Solution of nonlinear systems of equations} \label{sec:schemes}
We design and analyze quantum variants of two popular classical methods for solving the nonlinear equation \eqref{eq:nonlinear-eq}, namely fixed-point iteration and Newton's method. We show how they can be used in the context of the framework of Section~\ref{sec:framework} and characterize the runtime of the complete algorithms in terms of error tolerance. Importantly, while Lemma~\ref{lem:composition} implies that the runtime of either method is exponential in the number of steps, the total runtime is not exponential in terms of the tolerance $\tol$, since only $\log \tol^{-1}$ and $\log\log\tol^{-1}$ many steps are needed, respectively. 

Both methods iteratively generate a series of points~$(x^{(n)})_{n \in \N}$ that converge to a solution~$x^*$ under suitable conditions. Each iterate~$x^{(n+1)}$ depends only on the previous one~$x^{(n)}$, starting with an initial value~$x^{(0)}$. In practice, the complexity of applying some encoding~$\F$ of a function to~$x^{(n)}$ will not only be given by the quantity~$\qc(\F, x^{(n)})$ of Lemma~\ref{lem:composition} but also depend on the error propagated from the inexact previous computation of $x^{(n)}$. For the convergence results to hold, we thus have to assume boundedness of the complexity in the neighborhood of relevant points.
This condition is satisfied for encodings constructed using Theorem~\ref{thm:block-encoding-poly} except for $f(x) = 0$, but we will see that this is not problematic unless $x^{(n)} = 0$ for some $n \in \N$. In this case, the computation could be restarted using $x^{(0)} = 0$.
However, note that we do not necessarily assume $f$ to be a polynomial. Lemma~\ref{lem:invert} could, for example, be used to construct an encoding of $x \mapsto 1/x$.

\subsection{Fixed-point iteration}
The fixed-point iteration of a function $g \colon \C^N \to \C^N$ and initial value $x^{(0)} \in \C^N$ does not aim to solve an equation of the form \eqref{eq:nonlinear-eq} but rather to find the fixed-point $x^* \in \C^N$ of $g$ satisfying $x^* = g(x^*)$. Both problems are equivalent by setting $g(x) \coloneqq f(x) + x$. The fixed-point iteration computes
\[
x^{(n)} = g(g(\dots(g(x^{(0)})))) = g^n(x^{(0)})
\]
for $n \in \N$. The iterates $x^{(n)}$ are known to approach the solution of \eqref{eq:nonlinear-eq} under certain assumptions, for example when $g$ restricts to a contraction $X \to X$ with a closed set $X$ that contains $x^{(0)}$ and a contraction constant $0 \le L < 1$. In this case the fixed-point is also guaranteed to exist and be unique, and the absolute error the $n$-th step can be bounded by
\[
e_n \coloneqq |x^{(n)} - x^*| \le L^n e_0.
\]
For simplicity, we assume that $e_0 \le |x^*|$. Otherwise, $x^{(0)} = 0$ would be a better initial guess. If a relative error bounded by $\tol_g > 0$ is made in each application of $g$, the error bound is slightly weakened to
\[ L^n e_0 + 2 (1 - L)^{-1} |x^*| \tol_g,\]
where $2|x^*|$ is used as an upper limit on the norm of the iterates.
To achieve an error that can be bounded by some relative tolerance $\tol > 0$, one needs to perform
\[
n_\tol = \left\lceil\frac{|\log (\tol - 2(1 - L)^{-1} \tol_g) + \log |x^*| - \log e_0|}{|\log L|}\right\rceil
\]
steps, assuming $\tol_g < \tfrac12 (1 - L)\tol$. The implementation of the fixed-point iteration for amplified encodings using the function composition of Lemma~\ref{lem:composition} is now straightforward.

\begin{theorem}[Quantum fixed-point iteration] \label{thm:fixpoint}
    Let $N \in \N$, and $g \colon \C^N \to \C^N$. Assume that we have access to an efficient amplified encoding $\G$ of $g$, as well as an efficient amplified encoding $\A$ of a vector $x^{(0)} \in \C^N$. We further assume: \begin{itemize}
        \item The function $g$ restricts to a contraction $X \to X$ with $X \subseteq \C^N$ and contraction constant $0 \le L < 1$, such that $x^{(0)} \in X$.
        \item The complexity of $\qc(\G, x)$ of $\G$ is uniformly bounded for $x \in X$ by some constant $\qc_\G \ge 1$.
    \end{itemize}
    Then we can find an efficient amplified encoding of the unique fixed-point $x^*$ with complexity
    \[
    \Cg \tol^{-|\log \qc_\G|/|\log L|}\qc(\A).
    \]
    The constant $\Cg > 0$ depends on $L$, $e_0$, and $|x^{(0)}|$, but is independent of $g$, $\tol$, and $\delta$.
\end{theorem}
\begin{proof}
Let $\tol > 0$ and $0 < \delta < 1$. Note that by choosing $\tol_g = \tfrac{1 - L}{4}\tol$ and $\delta_g = \delta/n_\tol$ we have $\tol - 2(1 - L)^{-1}\tol_g = \tfrac12 \tol$ and $n_\tol\delta_g = \delta$, so we know that the desired tolerance is reached after the $n_\tol$ steps with the desired probability.
We apply Lemma~\ref{lem:composition} iteratively. Together, we get a total complexity of\[
\qc(\A) \prod_{n=0}^{n_\tol-1} \qc\big(\G, \tilde x^{(n)}\big) \le \qc(\A) \qc_\G^{n_\tol} \le  \qc(\A) \Cg\tol^{-|\log \qc_\G|/|\log L|}\]
where $\tilde x^{(n)}$ is the actual computed iterate including previous errors and
\[\Cg \coloneqq \qc_\G^{|\log |x^*| - \log e_0 - \log 2|/|\log L|}.
\]
Similarly, it is easy to see that the number of gates executed in the construction algorithm is at most $n_\tol$ times the gates in $\G$. Since $n_\tol \in \mathcal{O}(\log \tol^{-1}) \subset \mathcal{O}(\tol^{-|\log\qc_\G|/|\log L|})$ this implies that the described amplified encoding is efficient.
\end{proof}

Including the measurement by Lemma~\ref{lem:measure}, the runtime of the quantum fixed point iteration scales as $\tol^{-s} \log \delta^{-1}$, where $s > 1$ grows logarithmically with $\qc_\mathcal{G}$ and thus logarithmically with the degree of $g$ (according to Theorem~\ref{thm:block-encoding-poly}). To achieve a minimum runtime that scales linearly with the reciprocal of the error tolerance ($1/\tol$), a faster convergence method is required, such as Newton's method discussed in the next section.

\begin{remark}[Quantum advantage in high dimension]
The quantum fixed-point iteration is slower than the linear convergence of the classical fixed-point iteration with respect to the reciprocal error tolerance $\tol^{-1}$. However, this observation does not explicitly take into account the dependence on the potentially high dimension $N$. For $N \to \infty$, this dimensional dependence could provide a quantum advantage: the quantum runtime and memory requirements related to $N$ could scale exponentially smaller than its classical counterpart. As a result, for sufficiently high-dimensional problems, the quantum fixed-point iteration may still retain an overall quantum advantage.
\end{remark}

\subsection{Newton's method}\label{ss:newton}
Newton's method is certainly the most popular scheme to solve nonlinear equations, and computes iterates \begin{equation} \label{eq:newton-original}    
x^{(n+1)} = x^{(n)} - \Diff f\big(x^{(n)}\big)^{-1}f\big(x^{(n)}\big),\;n=0,1,2,\ldots,
\end{equation}
given some initial guess $x^{(0)}$. To realize this method in the present framework for quantum computing, the subsequent equivalent iteration rule 
\begin{equation} \label{eq:newton-alternative}
x^{(n+1)} = \Diff f\big(x^{(n)}\big)^{-1}\big(\Diff f\big(x^{(n)}\big)x^{(n)} - f\big(x^{(n)}\big)\big)
\end{equation}
is advantageous. Since $\Diff f(x^{(n)})x^{(n)} - f(x^{(n)})$ is itself a polynomial with the same coefficients as the polynomial $f$ up to explicit scalar factors, direct evaluations of $f(x^{(n)})$ can be avoided. According to Theorem~\ref{thm:block-encoding-poly}, they become delicate because $f(x^{(n)})$ necessarily approaches zero in the case of convergence.
The term $\Diff f(x^{(n)})x^{(n)} - f(x^{(n)})$ will also stay away from zero for all $n$, if $x^* \neq 0$ and $\Diff f(x^*)$ is invertible. The latter assumption is needed for Newton's method anyway. To quantify convergence properties, we specifically assume that the condition number $\kappa(\Diff f(\bullet))$ is bounded by some $\kappa > 0$. This bound always exists under suitable assumptions, as long as $\Diff f(\bullet)$ is invertible at relevant points.

Newton's method has \textit{local quadratic convergence}. The locality in this case means that the initial value $x^{(0)}$ should be sufficiently close to a solution $x^*$ to \eqref{eq:nonlinear-eq} in the sense that the initial error $e_0 < e_0^{\max}$, where $e_0^{\max} > 0$ is some constant that depends on $f$. We define $\theta$ as the ratio $\theta \coloneqq e_0/e_0^{\max}$. The error of $e_n$ after $n$ steps then satisfies the estimate
\[
e_n \le \theta^{2^n} e_0^{\max}.
\]
To reach a tolerance of $\tol$ the needed steps can be bounded by
\begin{align*}
n_\tol &= \log(\log \tol^{-1} + \log e_0^{\max}) - \log\log \theta^{-1}.
\end{align*}
This assumes that the new iterate is calculated with perfect accuracy.
In addition to the inaccuracy, which is introduced through the use of amplitude amplification, we also have to choose the accuracy of the linear solver, see Lemma~\ref{lem:invert}.
According to~\cite{DES82}, it is sufficient to set the relative tolerance of the solver to $2^{-n}$ in step $n$ to guarantee the local quadratic convergence $n_\tol \in \mathcal{O}(\log\log \tol^{-1})$. This refers to the original formulation~\eqref{eq:newton-original} of Newton's method, where the solver tolerance is relative to $|f(x^{(n)})|$. These values converge to zero, so the \emph{absolute} tolerance is much smaller in the limit $n \to \infty$. For the alternative formulation~\eqref{eq:newton-alternative}, we set the relative solver tolerance to
\[
\frac{e_n}{\kappa|x^{(n+1)}|}2^{-n} \leq\frac{|\Diff f(x^{(n)})^{-1}f(x^{(n)})|}{|x^{(n+1)}|} 2^{-n}.
\]
For our analysis, we are interested in the error after $n = n_\tol$ steps. For simplicity, we will use the same solver tolerance $C_\kappa 2^{-n_\tol} \tol$ for all $n_\tol$. With $C_\kappa = \kappa^{-1} \min_n|x^{(n)}|^{-1}$ independent of $\tol$, this choice is smaller than the individual solver tolerances above required to ensure quadratic convergence. The extra factor of $\tol$ here could potentially be prevented by using linear solvers adapted to this scenario.

\begin{theorem}[Quantum  implementation of Newton's method] \label{thm:newton}
    Let $N \in N$, and $f \colon \C^N \to \C^N$ with some efficient amplified encoding~$\F$ and an exact constructive encoding~$\F'$ of 
    \[ x \mapsto \Diff f(x)x - f(x) \qquad\text{and}\qquad x \mapsto \Diff f(x) \]
    respectively as well as an efficient amplified encoding $\A$ of a vector $x^{(0)} \in \C^N$. Let $x^* \in \C^N$ be the solution to the nonlinear equation~\eqref{eq:nonlinear-eq} and consider the set $X = \{ x \in \C^N \mid |x - x^*| \le |x - x^{(0)}|\}$ of points closer to $x^*$ than the initial guess. We assume further
    \begin{itemize}
        \item The function $f$ is twice continuously differentiable, and the Jacobian $Df(x)$ is invertible for all $x \in X$.
        \item The error of the initial guess satisfies $e_0 < e_0^{\max}$, where $e_0^{\max}$ is the constant related to $f$ introduced above.
        \item The complexities of $\F$ and $\F'$ in $X$ are bounded by constants $\qc_\F, \qc_\F' \ge 1$.
    \end{itemize}
    Then we can find an efficient amplified encoding of $x^*$ with ``subpolynomial'' complexity, i.e.,~for any $\mu > 0$ there exists a constant $C_\mu > 0$ depending on $e_0$, $e_0^{\max}$, the complexity of the encodings, and $\kappa$, such that the complexity is bounded by $C_\mu \tol^{-\mu}$.
\end{theorem}
\begin{proof}
    One Newton step consists of composing the matrix inverse of $\F'$ and $\F$. To have an error of at most $C_\kappa 2^{-n_\tol} \tol$ in each iteration, we run $\F$ with tolerance $\tfrac14 C_\kappa 2^{-n_\tol} \tol$ and solve the linear system using Lemma~\ref{lem:invert} with tolerance $\tfrac12 C_\kappa 2^{-n_\tol} \tol$. The total complexity of computing $x^{(n_\tol)}$ is then bounded by
    \begin{align}
    &\prod_{n=0}^{n_\tol-1} C_\mathrm{sol} \kappa \log(2 C_\kappa^{-1} 2^{n_\tol}\tol^{-1}) \big(\qc\big(\F, \tilde x^{(n)}\big) + \qc\big(\F', \tilde x^{(n)}\big)\big) \notag \\
    &\qquad \le \big(C_\mathrm{sol} \kappa \underbrace{\log(2 C_\kappa^{-1} 2^{n_\tol}\tol^{-1})}_{\log(2C_\kappa^{-1}) + n_\tol + \log \tol^{-1}} (\qc_\F + \qc_F')\big)^{n_\tol}
    \le (C(1 + n_\tol + \log\tol^{-1}))^{n_\tol}, \label{eq:newton:complexity}
    \end{align}
    where $\tilde x^{(n)}$ are the actual iterates including previous errors, $C_\mathrm{sol}$ is the hidden constant in Lemma~\ref{lem:invert}, and
    \[
    C \coloneqq C_\mathrm{sol}\kappa\max\{\log(2C_\kappa^{-1}), 1\}(\qc_\F + \qc_\F').
    \]
    Since $n_\tol$ is of order $\log\log \tol^{-1}$, thus growing slower than $\log \tol^{-1}$, the complexity~\eqref{eq:newton:complexity} asymptotically behaves as
    \begin{equation} \label{eq:newton:complexity2}
    \mathcal{O}((C_1\log \tol^{-1})^{C_2\log\log \tol^{-1}}) = \mathcal{O}\left((\log\tol^{-1})^{C_2\log C_1}\Big((\log \tol^{-1})^{\log\log \tol^{-1}}\Big)^{C_2}\right)
    \end{equation}
    with some constants $C_1, C_2 > 0$. Clearly, the first factor is $\polylog \tol^{-1}$, thus asymptotically smaller than $\tol^{-\mu/2}$ for any $\mu > 0$. We recall the fact that $\mathcal{O}(y^{\log y}) \subset \mathcal{O}(a^y)$ for arbitrarily small $a > 1$. By inserting $y = \log \tol^{-1}$ and $a = 2^{\mu/2C_2}$ we see that the second factor in~\eqref{eq:newton:complexity2} is also bounded by
    \[
    \big(a^{\log \tol^{-1}}\big)^{C_2} = \tol^{-C_2 \log a} = \tol^{-\mu/2}.
    \] 
    Construction efficiency can be shown analogously to Theorem~\ref{thm:fixpoint}.
\end{proof}

The main result of this paper, Theorem~\ref{thm:main}, then follows by combining it with Theorem~\ref{thm:block-encoding-poly} and Lemma~\ref{lem:measure} in a straightforward way. Recall that Theorem~\ref{thm:main} establishes that the theoretical complexity of our quantum implementation of Newton's method for polynomials is nearly linear with respect to $\tol^{-1}$. As mentioned earlier, this complexity result is not unique to our approach, but can be achieved for any tensor state or QSVT-based method -- potentially requiring the introduction of additional amplitude amplification as laid out in Section~\ref{ssec:amplification}.

\begin{remark}[Near optimal complexity]
    To measure any quantity of interest related to the amplitudes of a quantum state with error tolerance \(\tol\), it is well known that the process generating the state must be repeated a number of times proportional to \(\tol^{-1}\)~\cite{Md23}. Consequently, it is reasonable to assume that any quantum nonlinear solver that exploits the exponential speedup in \(N\) from amplitude-based data encoding must also have a complexity that scales at least linearly with \(\tol^{-1}\).
    To formalise this, note that the formulation of nonlinear systems of equations as presented here includes linear equations. Thus, the known lower bound of $\mathcal{O}(\kappa \log \tol^{-1})$ for preparing a solution state for linear systems, see e.g.~\cite{LS24}, must hold for this class of problems as well. To our knowledge, there is no approach that circumvents the $\tol^{-1}$ measurement cost for the linear case.
    This observation underscores that the overall complexity stated in Theorem~\ref{thm:main} is nearly optimal with respect to \(\tol\), at least while requiring polynomial dependency on $\kappa$ and $\log N$. However, whether the overhead of \(\tol^{-\mu}\) can be reduced to \(\log \tol^{-1}\) or even to a constant factor remains an open question.
\end{remark}

\section{Numerical experiments} \label{sec:numerics}
We test the general framework of Section~\ref{sec:framework} and the nonlinear solvers of Section~\ref{sec:schemes}, both on simulators and real hardware. The complete code written in the Qiskit framework \cite{JTK+24} is available in \url{https://github.com/MDeiml/quantum-nonlinear}, which we also refer to for a comprehensive explanation of all the technical details.

We consider the $2$-dimensional function
\[
g(x) = \begin{bmatrix} 1 \\ 1 \end{bmatrix} - \frac18 \begin{bmatrix} (x_1 + x_2)^2 \\ (x_1 - x_2)^2 \end{bmatrix}.
\]
This is a multivariate polynomial that can be cast in the format of Theorem~\ref{thm:block-encoding-poly} with 
\[
A_0 = \begin{bmatrix} 1 \\ 1 \end{bmatrix}, \qquad
A_1 = 0, \qquad
A_2 = -\frac18 \begin{bmatrix} 1 & \phantom{-}1 & \phantom{-}1 & 1 \\ 1 & -1 & -1 & 1 \end{bmatrix}.
\]
However, we can also implement an amplified encoding of $g$ directly following its definition. For the quadratic part of $g$ we then construct two copies of $x$, apply the unitary matrix
\[
\mathrm{H} = \frac{1}{\sqrt{2}}\begin{bmatrix}
    1 & \phantom{-}1 \\ 1 & -1
\end{bmatrix}
\]
and multiply them as in Example~\ref{ex:mul}. The circuit for this process is sketched in Figure~\ref{fig:circuit}. 
\begin{figure}
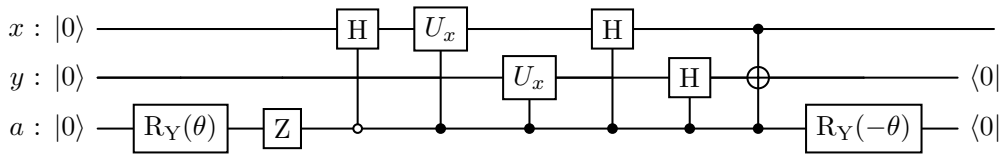

    \centering
    \includestandalone{circuit}
    \caption{Constructive encoding implementing $g(x)$. The angle $\theta$ is related to the addition and depends on the norm of $x$, specifically $\theta = 2 \tan^{-1} |x|/\sqrt[4]{2}$.}
    \label{fig:circuit}
\end{figure}
We can check using simulators that by doing this construction for $U_x$ which prepares $x^{(0)} = [1 \; 1]^T$ we actually get a block encoding of $g(x^{(0)})$.

We now attempt to compute the fixed-point $x^*$ of $g$. Numerically one can check that one of two fixed-points of $g$ is given by $x^* \approx [0.66\; 0.99]^T$. The Jacobian $Dg(x^*)$ at this fixed-point has a spectral norm of roughly $L = 0.58$, so we know that the fixed-point iteration converges if the initial guess is close enough. In fact, we can numerically verify that for the initial guess $x^{(0)} = [1\; 1]$, the error after the $4$ steps is roughly $4 \cdot 10^{-3}$. Following Theorem~\ref{thm:fixpoint} we iteratively construct block encodings $U_{x^{(n)}}$, estimate their information efficiency, and amplify them correspondingly. Using simulated quantum devices, we see that the actual iterates $x^{(n)}$ are encoded almost exactly, see also Figure~\ref{fig:fixed-point-iterations}. This means that $U_{x^{(4)}}$ encodes the fixed-point $x^*$ to a tolerance of $4 \cdot 10^{-3}$.
Naturally, the number of atomic gates in each iteration increases exponentially with the number of steps.
We also run the experiment on the \texttt{ibm\_fez} quantum computer, where we see that the first and second iterate are well approximated. Computation of the third point fails with an unspecified error, possibly because the corresponding circuit is too large. Since previous efforts in the literature did not yield such practical results, this still marks an important step towards solving nonlinear problems on quantum computers.

\begin{figure}
    \centering
    \begin{tikzpicture}
        \begin{axis}[
		width=0.42\textwidth,
		height=0.21\textwidth,
		at={(-0.05\textwidth,0\textwidth)},
		scale only axis,
		unbounded coords=jump,
        domain=0.001:1,
        xmin=0.43,
        xmax=1.03,
        ymin=0.88,
        ymax=1.18,
        legend columns=4,
		legend style={
            /tikz/every even column/.append style={column sep=1ex}
		},
        ]
        \addplot[color=black, mark=o, only marks, thick] coordinates {(0.66, 0.99)};
        \addplot[color=mycolor4, mark=+, thick, mark size=4pt, dashed, mark options={solid}] coordinates {(1.0,1.0) (0.5,1.0) (0.71875, 0.96875) (0.64404297, 0.9921875)};
        \addplot[color=mycolor1, mark=x, thick, mark size=4pt, only marks] coordinates {(1.0,1.0) (0.4999357,  0.99951845) (0.72085069, 0.96736464) (0.64814715, 0.99072392)};
        \addplot[color=mycolor2, mark=square, thick, mark size=4pt, only marks, mark options={solid}] coordinates {(1.0,1.0) (0.46865308, 1.06265909) (0.74254395, 0.97134777)};
        \node (a) at (axis cs:1,1) [anchor=south,inner sep=8pt,node font=\scriptsize] {$x^{(0)}$};
        \draw[->] (axis cs:0.55,1.035) -- (axis cs:0.51,1.01);
        \draw[->] (axis cs:0.55,1.05) -- (axis cs:0.49, 1.06265909);
        \node (a) at (axis cs:0.55,1.05) [anchor=west,inner sep=2pt,node font=\scriptsize] {$x^{(1)}$};
        \node (a) at (axis cs:0.75,0.97) [anchor=north,inner sep=8pt,node font=\scriptsize] {$x^{(2)}$};
        \node (a) at (axis cs:0.65,0.98) [anchor=south,inner sep=8pt,node font=\scriptsize] {$x^{(3)}$};
        \end{axis}
        \begin{axis}[
		width=0.42\textwidth,
		height=0.21\textwidth,
		at={(0.46\textwidth,0\textwidth)},
		scale only axis,
		unbounded coords=jump,
        xmin=1.9,
        xmax=3.1,
        ymin=-0.25,
        ymax=0.35,
        legend columns=4,
		legend style={
            /tikz/every even column/.append style={column sep=1ex},
            at={(0.8,1.1)},
            anchor=south east,
		},
        ]
        \addplot[color=black, mark=o, only marks, thick] coordinates {(2.82842712475, 0)};
        \addlegendentry{fixed-point (left)/root (right)}
        \addplot[color=mycolor4, mark=+, thick, mark size=4pt, dashed, mark options={solid}] coordinates {(2, 0.25) (3.032, -0.129) (2.838, -0.008)};
        \addlegendentry{reference}
        \addplot[color=mycolor1, mark=x, thick, mark size=4pt, only marks] coordinates {(2, 0.25) (3.028,  -0.125) (2.837, -0.008)};
        \addlegendentry{noiseless}
        \addplot[color=mycolor2, mark=square, thick, mark size=4pt, only marks, mark options={solid}] coordinates {(2,0.25) (2.91955146, 0)};
        \addlegendentry{\texttt{ibm\_\{fez,aachen\}}}
        \node (a) at (axis cs:2,0.25) [anchor=north,inner sep=8pt,node font=\scriptsize] {$x^{(0)}$};
        \draw[->] (axis cs:3.032,0.1) -- (axis cs:3.032,-0.1);
        \draw[->] (axis cs:3.0,0.1) -- (axis cs:2.96, 0.04);
        \node (a) at (axis cs:3.032,0.1) [anchor=south,inner sep=2pt,node font=\scriptsize] {$x^{(1)}$};
        \draw[->] (axis cs:2.75,0.1) -- (axis cs:2.81, 0.02);
        \node (a) at (axis cs:2.75,0.1) [anchor=south,inner sep=2pt,node font=\scriptsize] {$x^{(2)}$};
        \end{axis}
    \end{tikzpicture}
    \caption{First three fixed-point iterations of $g$ with starting value $\big[1\; 1\big]^T$ (left) and first two Newton iterations with starting value $\big[2\; \tfrac14\big]^T$ (right). The plot includes classically computed reference values~($+$), values computed using a noiseless quantum simulator ($\times$) and the quantum computers \texttt{ibm\_fez} and \texttt{ibm\_aachen} ($\square$) available through IBM Cloud.}
    \label{fig:fixed-point-iterations}
\end{figure}
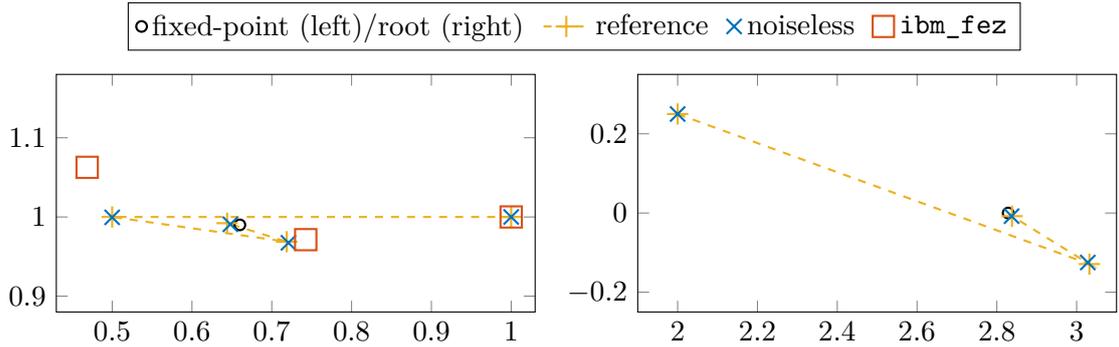
Figure~\ref{fig:fixed-point-iterations} also contains results for the Newton's method applied to $g$ with the initial value~$\big[2\; \tfrac14\big]^T$. Note that this approximates a root of $g$ rather than a fixed-point. The quantum circuits for Newton's method are significantly larger than those for fixed-point iteration, since they require the implementation of a linear solver. For the linear solver, we use a method based on QSVT \cite[Theorem 41]{GSLW19}, with a predefined error tolerance of \(\tol = 0.1\) and condition number bound $\kappa = 6$. The corresponding angles are computed with the tool \texttt{pyqsp} \cite{MRTC21}. The simulated quantum method matches the classical method closely and both approximate the root \(\big[2\sqrt{2}\; 0\big]^T\) with an accuracy of \(10^{-2}\) after only two steps. The first of these Newton steps was also executed on the \texttt{ibm\_aachen} quantum computer using a less precise linear solver, giving a reasonable approximation. The circuit for the second step is larger by a factor of $\sim40$, which prevented a corresponding hardware run.

Although the numerical example above illustrates the feasibility of our method in a low-dimensional setting, extending it to general high-dimensional problems presents new challenges and opportunities. In particular, the number of qubits required scales approximately with \(\log N\) times the number of steps for both Newton's method and fixed-point iterations. Due to the parallel nature of the tensor state construction, the circuits could also be adapted to achieve a favorable trade-off between depth and width. A simple calculation shows that the quantum advantage is particularly apparent with respect to memory. Specifically, the largest current supercomputers have memory to store about \(2^{46}\) floating-point numbers, corresponding to the state space of a quantum computer with about 46 qubits. Although quantum computers of this size already exist, our implementation of Newton's method remains infeasible. For example, even in our small 2-dimensional example, the number of gates required exceeds \(10^5\) per sample, all of which would have to be executed coherently. Such circuits require significant advances in hardware, especially in error correction.

On the other hand, our method is still open to many optimizations, for instance:
\begin{itemize}
    \item Combining amplitude amplification and amplitude estimation using maximum likelihood estimation~\cite{SUR+20} into one adaptive algorithm.
    \item Trading complexity for construction time, possibly through the development of tailored linear solvers.
    \item Finding polynomial decompositions, where the information efficiency is larger than in Theorem~\ref{thm:block-encoding-poly}.
    \item Replacing $x^{(n)}$ by earlier iterates in certain places of Newton's method, similar to~\cite{BC23}.
\end{itemize}
We hope to explore these approaches in the future.

\section{Conclusion}
We have presented a novel framework for linear and nonlinear numerical computation on quantum computers that is both general and practical. Our theoretical analysis shows that this framework can solve nonlinear systems of equations with a logarithmic dependence on problem size, suggesting a potentially exponential speedup over classical methods. Moreover, experiments on currently accessible quantum hardware validate the practicality of our approach, bringing us closer to solving relevant nonlinear problems on quantum computers. This work highlights the potential for achieving a quantum advantage, particularly in tackling high-dimensional nonlinear partial differential equations, and paves the way for relevant applications in science and engineering.

\section*{Acknowledgements}
The work of D. Peterseim is part of a project that has received funding from the European Research Council (ERC) under the European Union’s Horizon 2020 research and innovation programme (Grant agreement No. 865751 -- RandomMultiScales).

We acknowledge the use of IBM Quantum Credits for this work. The
views expressed are those of the authors, and do not reflect the official
policy or position of IBM or the IBM Quantum team.

\bibliographystyle{quantum}
\bibliography{p_quantum_Newton}

\begin{thebibliography}{10}

\bibitem{TLLM24}
Felix Tennie, Sylvain Laizet, Seth Lloyd, and Luca Magri.
\newblock ``Quantum {{Computing}} for nonlinear differential equations and
  turbulence''~(2024).
\newblock  \href{http://arxiv.org/abs/2406.04826}{arXiv:2406.04826}.

\bibitem{JL24}
Shi Jin and Nana Liu.
\newblock ``Quantum algorithms for nonlinear partial differential equations''.
\newblock \href{https://dx.doi.org/10.1016/j.bulsci.2024.103457}{Bulletin des
  Sciences Math{\'e}matiques {\bf 194}, 103457}~(2024).

\bibitem{JLY23a}
Shi Jin, Nana Liu, and Yue Yu.
\newblock ``Time complexity analysis of quantum algorithms via linear
  representations for nonlinear ordinary and partial differential equations''.
\newblock \href{https://dx.doi.org/10.1016/j.jcp.2023.112149}{Journal of
  Computational Physics {\bf 487}, 112149}~(2023).
\newblock  \href{http://arxiv.org/abs/2209.08478}{arXiv:2209.08478}.

\bibitem{LJM+20}
Michael Lubasch, Jaewoo Joo, Pierre Moinier, Martin Kiffner, and Dieter Jaksch.
\newblock ``Variational quantum algorithms for nonlinear problems''.
\newblock \href{https://dx.doi.org/10.1103/PhysRevA.101.010301}{Physical Review
  A {\bf 101}, 010301}~(2020).
\newblock  \href{http://arxiv.org/abs/1907.09032}{arXiv:1907.09032}.

\bibitem{LDG+20}
Seth Lloyd, Giacomo De~Palma, Can Gokler, Bobak Kiani, Zi-Wen Liu, Milad
  Marvian, Felix Tennie, and Tim Palmer.
\newblock ``Quantum algorithm for nonlinear differential equations''~(2020).
\newblock  \href{http://arxiv.org/abs/2011.06571}{arXiv:2011.06571}.

\bibitem{HCSS23}
Zo{\"e} Holmes, Nolan~J. Coble, Andrew~T. Sornborger, and Yi{\u g}it Suba{\c
  s}{\i}.
\newblock ``Nonlinear transformations in quantum computation''.
\newblock \href{https://dx.doi.org/10.1103/PhysRevResearch.5.013105}{Physical
  Review Research {\bf 5}, 013105}~(2023).
\newblock  \href{http://arxiv.org/abs/2112.12307}{arXiv:2112.12307}.

\bibitem{LO08}
Sarah~K. Leyton and Tobias~J. Osborne.
\newblock ``A quantum algorithm to solve nonlinear differential
  equations''~(2008).
\newblock  \href{http://arxiv.org/abs/0812.4423}{arXiv:0812.4423}.

\bibitem{RSW+19}
Patrick Rebentrost, Maria Schuld, Leonard Wossnig, Francesco Petruccione, and
  Seth Lloyd.
\newblock ``Quantum gradient descent and {{Newton}}'s method for constrained
  polynomial optimization''.
\newblock \href{https://dx.doi.org/10.1088/1367-2630/ab2a9e}{New Journal of
  Physics {\bf 21}, 073023}~(2019).
\newblock  \href{http://arxiv.org/abs/1612.01789}{arXiv:1612.01789}.

\bibitem{GLW+21}
Pan Gao, Keren Li, Shijie Wei, Jiancun Gao, and Guilu Long.
\newblock ``Quantum gradient algorithm for general polynomials''.
\newblock \href{https://dx.doi.org/10.1103/PhysRevA.103.042403}{Physical Review
  A {\bf 103}, 042403}~(2021).
\newblock  \href{http://arxiv.org/abs/2004.11086}{arXiv:2004.11086}.

\bibitem{GLWL21}
Pan Gao, Keren Li, Shijie Wei, and Gui-Lu Long.
\newblock ``Quantum second-order optimization algorithm for general
  polynomials''.
\newblock \href{https://dx.doi.org/10.1007/s11433-021-1725-9}{Science China
  Physics, Mechanics \& Astronomy {\bf 64}, 100311}~(2021).

\bibitem{LWG+21}
Keren Li, Shijie Wei, Pan Gao, Feihao Zhang, Zengrong Zhou, Tao Xin, Xiaoting
  Wang, Patrick Rebentrost, and Guilu Long.
\newblock ``Optimizing a polynomial function on a quantum processor''.
\newblock \href{https://dx.doi.org/10.1038/s41534-020-00351-5}{npj Quantum
  Information {\bf 7}, 16}~(2021).
\newblock  \href{http://arxiv.org/abs/1804.05231}{arXiv:1804.05231}.

\bibitem{CGWL24}
Lei Cheng, Pan Gao, Tiejun Wang, and Keren Li.
\newblock ``Polynomial optimization with linear combination of unitaries''.
\newblock \href{https://dx.doi.org/10.1103/PhysRevA.109.032429}{Physical Review
  A {\bf 109}, 032429}~(2024).

\bibitem{GMF21}
Naixu Guo, Kosuke Mitarai, and Keisuke Fujii.
\newblock ``Nonlinear transformation of complex amplitudes via quantum singular
  value transformation''~(2021).
\newblock  \href{http://arxiv.org/abs/2107.10764}{arXiv:2107.10764}.

\bibitem{RR23}
Arthur~G. Rattew and Patrick Rebentrost.
\newblock ``Non-{{Linear Transformations}} of {{Quantum Amplitudes}}:
  {{Exponential Improvement}}, {{Generalization}}, and
  {{Applications}}''~(2023).
\newblock  \href{http://arxiv.org/abs/2309.09839}{arXiv:2309.09839}.

\bibitem{GSLW19}
Andr{\'a}s Gily{\'e}n, Yuan Su, Guang~Hao Low, and Nathan Wiebe.
\newblock ``Quantum singular value transformation and beyond: Exponential
  improvements for quantum matrix arithmetics''.
\newblock In Proceedings of the 51st {{Annual ACM SIGACT Symposium}} on
  {{Theory}} of {{Computing}}.
\newblock \href{https://dx.doi.org/10.1145/3313276.3316366}{Pages 193--204}.
\newblock ~(2019).
\newblock  \href{http://arxiv.org/abs/1806.01838}{arXiv:1806.01838}.

\bibitem{NW24}
Nhat~A. Nghiem and Tzu-Chieh Wei.
\newblock ``Quantum {{Algorithm For Solving Nonlinear Algebraic
  Equations}}''~(2024).
\newblock  \href{http://arxiv.org/abs/2404.03810}{arXiv:2404.03810}.

\bibitem{Amb10}
Andris Ambainis.
\newblock ``Variable time amplitude amplification and a faster quantum
  algorithm for solving systems of linear equations''~(2010).
\newblock  \href{http://arxiv.org/abs/1010.4458}{arXiv:1010.4458}.

\bibitem{CGJ19}
Shantanav Chakraborty, Andr{\'a}s Gily{\'e}n, and Stacey Jeffery.
\newblock ``The power of block-encoded matrix powers: Improved regression
  techniques via faster {{Hamiltonian}} simulation''.
\newblock \href{https://dx.doi.org/10.4230/LIPIcs.ICALP.2019.33}{LIPIcs, Volume
  132, ICALP 2019 {\bf 132}, 33:1--33:14}~(2019).
\newblock  \href{http://arxiv.org/abs/1804.01973}{arXiv:1804.01973}.

\bibitem{LS24}
Guang~Hao Low and Yuan Su.
\newblock ``Quantum linear system algorithm with optimal queries to initial
  state preparation''~(2024).
\newblock  \href{http://arxiv.org/abs/2410.18178}{arXiv:2410.18178}.

\bibitem{RNF+23}
Mehdi Ramezani, Morteza Nikaeen, Farnaz Farman, Seyed~Mahmoud Ashrafi, and
  Alireza Bahrampour.
\newblock ``Quantum {{Multiplication Algorithm Based}} on the {{Convolution
  Theorem}}''.
\newblock \href{https://dx.doi.org/10.1103/PhysRevA.108.052405}{Physical Review
  A {\bf 108}, 052405}~(2023).
\newblock  \href{http://arxiv.org/abs/2306.08473}{arXiv:2306.08473}.

\bibitem{GYC+24}
Naixu Guo, Zhan Yu, Matthew Choi, Aman Agrawal, Kouhei Nakaji, Al{\'a}n
  {Aspuru-Guzik}, and Patrick Rebentrost.
\newblock ``Quantum linear algebra is all you need for {{Transformer}}
  architectures''~(2024).
\newblock  \href{http://arxiv.org/abs/2402.16714}{arXiv:2402.16714}.

\bibitem{Zyl24}
Julien Zylberman.
\newblock ``Fast {{Laplace}} transforms on quantum computers''~(2024).
\newblock  \href{http://arxiv.org/abs/2412.05173}{arXiv:2412.05173}.

\bibitem{BH96}
Vladimir Bu{\v z}ek and Mark Hillery.
\newblock ``Quantum copying: {{Beyond}} the no-cloning theorem''.
\newblock \href{https://dx.doi.org/10.1103/PhysRevA.54.1844}{Physical Review A
  {\bf 54}, 1844--1852}~(1996).
\newblock
  \href{http://arxiv.org/abs/quant-ph/9607018}{arXiv:quant-ph/9607018}.

\bibitem{CKS17}
Andrew~M. Childs, Robin Kothari, and Rolando~D. Somma.
\newblock ``Quantum {{Algorithm}} for {{Systems}} of {{Linear Equations}} with
  {{Exponentially Improved Dependence}} on {{Precision}}''.
\newblock \href{https://dx.doi.org/10.1137/16M1087072}{SIAM Journal on
  Computing {\bf 46}, 1920--1950}~(2017).

\bibitem{DP25}
Matthias Deiml and Daniel Peterseim.
\newblock ``Quantum {{Realization}} of the {{Finite Element Method}}''.
\newblock \href{https://dx.doi.org/10.1090/mcom/4124}{Math. Comp.}~(2025).
\newblock  \href{http://arxiv.org/abs/2403.19512}{arXiv:2403.19512}.

\bibitem{Gro98}
Lov~K. Grover.
\newblock ``Quantum {{Computers Can Search Rapidly}} by {{Using Almost Any
  Transformation}}''.
\newblock \href{https://dx.doi.org/10.1103/PhysRevLett.80.4329}{Physical Review
  Letters {\bf 80}, 4329--4332}~(1998).
\newblock
  \href{http://arxiv.org/abs/quant-ph/9712011}{arXiv:quant-ph/9712011}.

\bibitem{BH97}
Gilles Brassard and Peter Hoyer.
\newblock ``An exact quantum polynomial-time algorithm for {{Simon}}'s
  problem''.
\newblock In Proceedings of the {{Fifth Israeli Symposium}} on {{Theory}} of
  {{Computing}} and {{Systems}}.
\newblock \href{https://dx.doi.org/10.1109/ISTCS.1997.595153}{Pages 12--23}.
\newblock Ramat-Gan, Israel~(1997). IEEE Comput. Soc.
\newblock
  \href{http://arxiv.org/abs/quant-ph/9704027}{arXiv:quant-ph/9704027}.

\bibitem{SUR+20}
Yohichi Suzuki, Shumpei Uno, Rudy Raymond, Tomoki Tanaka, Tamiya Onodera, and
  Naoki Yamamoto.
\newblock ``Amplitude estimation without phase estimation''.
\newblock \href{https://dx.doi.org/10.1007/s11128-019-2565-2}{Quantum
  Information Processing {\bf 19}, 75}~(2020).
\newblock  \href{http://arxiv.org/abs/1904.10246}{arXiv:1904.10246}.

\bibitem{RF23}
Patrick Rall and Bryce Fuller.
\newblock ``Amplitude {{Estimation}} from {{Quantum Signal Processing}}''.
\newblock \href{https://dx.doi.org/10.22331/q-2023-03-02-937}{Quantum {\bf 7},
  937}~(2023).
\newblock  \href{http://arxiv.org/abs/2207.08628}{arXiv:2207.08628}.

\bibitem{AR20}
Scott Aaronson and Patrick Rall.
\newblock ``Quantum {{Approximate Counting}}, {{Simplified}}''.
\newblock In Symposium on {{Simplicity}} in {{Algorithms}}.
\newblock \href{https://dx.doi.org/10.1137/1.9781611976014.5}{Pages 24--32}.
\newblock Philadelphia, PA~(2020). {Society for Industrial and Applied
  Mathematics}.
\newblock  \href{http://arxiv.org/abs/1908.10846}{arXiv:1908.10846}.

\bibitem{AJL06}
Dorit Aharonov, Vaughan Jones, and Zeph Landau.
\newblock ``A {{Polynomial Quantum Algorithm}} for {{Approximating}} the
  {{Jones Polynomial}}''~(2006).
\newblock
  \href{http://arxiv.org/abs/quant-ph/0511096}{arXiv:quant-ph/0511096}.

\bibitem{DES82}
Ron~S. Dembo, Stanley~C. Eisenstat, and Trond Steihaug.
\newblock ``Inexact {{Newton Methods}}''.
\newblock \href{https://dx.doi.org/10.1137/0719025}{SIAM Journal on Numerical
  Analysis {\bf 19}, 400--408}~(1982).

\bibitem{Md23}
Nikhil~S. Mande and Ronald {de Wolf}.
\newblock ``Tight {{Bounds}} for {{Quantum Phase Estimation}} and {{Related
  Problems}}''~(2023).
\newblock  \href{http://arxiv.org/abs/2305.04908}{arXiv:2305.04908}.

\bibitem{JTK+24}
Ali {Javadi-Abhari}, Matthew Treinish, Kevin Krsulich, Christopher~J. Wood,
  Jake Lishman, Julien Gacon, Simon Martiel, Paul~D. Nation, Lev~S. Bishop,
  Andrew~W. Cross, Blake~R. Johnson, and Jay~M. Gambetta.
\newblock ``Quantum computing with {{Qiskit}}''~(2024).
\newblock  \href{http://arxiv.org/abs/2405.08810}{arXiv:2405.08810}.

\bibitem{MRTC21}
John~M. Martyn, Zane~M. Rossi, Andrew~K. Tan, and Isaac~L. Chuang.
\newblock ``A {{Grand Unification}} of {{Quantum Algorithms}}''.
\newblock \href{https://dx.doi.org/10.1103/PRXQuantum.2.040203}{PRX Quantum
  {\bf 2}, 040203}~(2021).
\newblock  \href{http://arxiv.org/abs/2105.02859}{arXiv:2105.02859}.

\bibitem{BC23}
Jeffrey~D. Blanchard and Marc Chamberland.
\newblock ``Newton's {{Method Without Division}}''.
\newblock \href{https://dx.doi.org/10.1080/00029890.2022.2093573}{The American
  Mathematical Monthly {\bf 130}, 606--617}~(2023).

\end{thebibliography}

\end{document}